\documentclass{article}
\usepackage{preprint}
\usepackage[utf8]{inputenc}\usepackage[T1]{fontenc}\usepackage{lmodern}
\usepackage{amsmath,amssymb,amsthm,mathtools,bm}
\usepackage{natbib}\usepackage[hidelinks]{hyperref}\usepackage{booktabs}\usepackage{graphicx}\usepackage{xcolor}\usepackage{enumitem}

\newtheorem{theorem}{Theorem}\newtheorem{proposition}[theorem]{Proposition}
\newtheorem{lemma}[theorem]{Lemma}\newtheorem{corollary}[theorem]{Corollary}
\theoremstyle{definition}\newtheorem{definition}[theorem]{Definition}
\theoremstyle{remark}\newtheorem{remark}[theorem]{Remark}
\newcommand{\indep}{\perp\!\!\!\perp}\newcommand{\doop}{\mathrm{do}}
\newcommand{\Tp}{\mathbb{T}^{+}}\newcommand{\Sel}{\mathsf{Sel}}
\newcommand{\Wg}{W_{\mathrm{gauge}}}\newcommand{\Wcw}{W_{\mathrm{cw}}}
\newcommand{\Dset}{\mathfrak{D}}\newcommand{\Sep}{\textsc{(sep)}}
\newcommand{\Sw}{\mathbb{S}}\newcommand{\Kt}{K^{\tau}_{\min}}
\newcommand{\sep}{\perp_\sigma}\newcommand{\An}{\mathrm{An}}
\newcommand{\GL}{\mathrm{GL}}

\title{Equilibrium Causal Digital Twins:\\ Validation, Transport,\\ and Identification Limits}
\author{
Faraz Dadgostari\\
Department of Mechanical \& Industrial Engineering, Montana State University
\and
Neda Nazemi\\
Gianforte School of Computing, Montana State University
}
\date{}
\begin{document}\maketitle

\begin{abstract}
Digital twins are often used to predict how a system would respond to an intervention. In systems with feedback, a
twin must reproduce an equilibrium counterfactual, and a twin developed in one domain may fail after mechanisms
change. We study when these predictions can be validated and transported. For equilibrium causal games, we give
conditions on the mechanisms, equilibrium selection, and intervention design under which agreement with
experimental distributions identifies the counterfactual of interest. We show why agreement of means and
covariances is insufficient for distributional queries. We then introduce cyclic selection diagrams and derive
criteria for direct reuse and for hybrid models that combine invariant source mechanisms with target information.
An impossibility result constructs systems that agree under every experiment in a finite design but disagree on the
target counterfactual, showing that validation requires structural assumptions. For linear models, we derive
intervention requirements that depend on the mechanisms that changed, the observation model, and graph support.
When point identification fails, we characterize the remaining range of query values. We also provide statistical
tests for reconstructed means and covariances and illustrate the theory in synthetic feedback systems.
\end{abstract}

\section{Introduction}
Digital twins are used to answer intervention questions: what would the system do if a mechanism were changed?
For a system with feedback, the answer is an equilibrium rather than a one-way propagation through an acyclic
graph. A simple two-agent feedback loop already shows the difference. Removing one feedback arrow can erase the
path that carries both the intervention effect and its sensitivity to a change of domain
\citep{Chin2021nonrobust}.

This paper separates three questions. \emph{Validation} asks whether a twin gives the correct counterfactual in
the domain where it was developed. \emph{Transport} asks whether that answer remains correct after some mechanisms
change. \emph{Identification} asks whether the available observational and experimental distributions determine
the relevant mechanisms or, more modestly, the query itself. These questions are related, but none can be replaced
by the others.

The distinction is especially important for counterfactuals. Interventions in cyclic latent models can be
identified under graphical and solvability conditions \citep{ForreMooij2019calculus,Bongers2021cyclic}, but a
counterfactual links factual evidence to a hypothetical intervention and therefore lies at a higher level of the
causal hierarchy \citep{BareinboimCII2022pch,Bongers2021cyclic}. In acyclic models, prior work characterizes
testable counterfactuals \citep{ShpitserPearl2007testable,TianPearl2002testable} and develops transportability
through selection diagrams and do-calculus
\citep{PearlBareinboim2014external,BareinboimPearl2016fusion,BareinboimPearl2013limited,
BareinboimPearl2012completeness,CorreaLeeBareinboim2022ctf}. Existing digital-twin approaches likewise use
acyclic potential-outcome models \citep{Laudy2026dtcf}, while time-unrolled approaches remain acyclic within each
time slice \citep{BlondelAriasGavalda2016dcn}. The equilibrium setting requires a different treatment because the
feedback solution itself carries the causal response.

\paragraph{Main results.}
First, we give conditions under which agreement with a collection of interventional distributions validates a
counterfactual within one domain. The conditions separate monotonicity of the mechanisms, stability of equilibrium
selection, and identification of the query-relevant response. They also distinguish full factual information from
partial factual information and distributional agreement from agreement of only means and covariances.

Second, we develop transport rules for cyclic selection diagrams. Direct reuse is valid when the relevant
post-intervention ancestors are invariant and aligned across domains. When changed mechanisms are relevant, a
hybrid model combines invariant source mechanisms with mechanisms re-identified in the target.

Third, we show why structural assumptions cannot be avoided. Under explicit witness conditions, two systems can
agree under every experiment in a finite design while assigning different values to the same counterfactual.
Thus experimental agreement alone does not validate a cross-world prediction.

Fourth, for linear models we derive target-intervention requirements that depend on the set of changed
mechanisms, the observation model, and graph support. We characterize the remaining set of query values when point
identification fails, show when query-specific experiments can require less information than full model recovery,
and provide statistical procedures for comparing reconstructed means and covariances.

\paragraph{Scope.}
The positive results require the structural and design conditions stated with each theorem. Intervention counts
are not universal: they depend on what is observed, which mechanisms may differ, and which support restrictions are
known. The finite-sample and asymptotic procedures likewise apply only under their stated sampling, rank, and
regularity assumptions. The numerical examples illustrate the constructions but do not establish general
identification or empirical performance.

\paragraph{Organization.}
Section~\ref{sec:obj} introduces equilibrium causal twins, counterfactual queries, and domain changes.
Sections~\ref{sec:lemma} and~\ref{sec:head} develop validation within one domain. Section~\ref{sec:transport}
gives direct and hybrid transport results, and Section~\ref{sec:semantics} explains why acyclic reductions can
fail. Section~\ref{sec:untest} establishes the finite-design impossibility result. Sections~\ref{sec:count},
\ref{sec:bounds}, and~\ref{sec:vmin} study target interventions, partial identification, alignment, and
query-specific design. Section~\ref{sec:instr} presents the statistical procedures, followed by numerical
illustrations and limitations.

\section{Equilibrium causal twins and counterfactual queries}\label{sec:obj}
\subsection{Model and equilibrium selection}
An equilibrium causal game (ECG) consists of structural equations
$V_i=f_i(V_{\mathrm{pa}(i)},U_i)$. The scalar noises are mutually independent and have continuous,
strictly increasing distribution functions. Each strongly connected component has a unique solution under a
declared measurable selection rule $\Sel$, and that rule is stable across the interventions under study. In the
linear specialization, $V=BV+c+U$, $\rho(B)<1$, and the observed variables satisfy $X=HV$, where $H$ has full
column rank and is invariant across environments.

\subsection{Counterfactual queries and interventions}
A per-unit equilibrium counterfactual $\chi(\mathcal S;\mathfrak f,I)$ is obtained by recovering the exogenous
state from the factual observation, applying intervention $I$, solving the modified system under $\Sel$, and
reading $e_Y^\top V_{cf}$. The query must be invariant under the residual gauge of the identified object (the
sign-absorbed orbit $\simeq_R$). Otherwise the gauge width in Section~\ref{sec:bounds} remains irreducible.

Unless stated otherwise, $I$ changes a mechanism. A payoff or curvature intervention is covered only after the
reward layer has been transported and re-evaluated at target moments. In particular, target omitted-variable
bias can reintroduce bias into a reward that was unbiased in the source, and the transported reward's
$\lambda$-scale gauge must be restored once per domain. These obligations are established in the companion
reward paper \citep{ECG3_2026}. The results below extend to payoff interventions only when those target-domain
conditions hold.

\subsection{Conditions for validation}
Validation proceeds through four logically separate steps. The first recovers factual noise ranks, the second
determines what partial factual information implies, the third fixes the equilibrium branch, and the fourth
asks whether the experiment identifies the mechanisms relevant to the query. Table~\ref{tab:boundary}
summarizes their roles; the formal conditions follow immediately below.
\begin{table}[t]
\centering\small
\begin{tabular}{@{}p{0.11\textwidth}p{0.27\textwidth}p{0.50\textwidth}@{}}
\toprule
Condition & Object controlled & Consequence for the query \\
\midrule
T1 & strict noise monotonicity & factual values recover the private-noise ranks \\
T2 & complete legal fibre & every legal member induces the same post-intervention query kernel \\
T3 & SCC-local selection & ranks and the solution set select the same equilibrium branch \\
T4 & declared response object & response equality forces query-kernel equality throughout the fibre \\
\bottomrule
\end{tabular}
\caption{The four steps in the validation argument. None substitutes for another: T1 concerns abduction, T2
partial factual information, T3 equilibrium choice, and T4 the information supplied by the experiment.}
\label{tab:boundary}
\end{table}

\paragraph{T1: noise ranks.}
For every node and parent value, $f_i(v,\cdot)$ is strictly increasing. A factual value therefore recovers the
rank of its private noise.

\begin{definition}[Partial-factual query condition (T2)]
For retained evidence $y$, let $\mathcal F_{\mathcal D}(y)$ be the complete legal fibre, let
$W$ be a common labelled, gauge-invariant standard-Borel factual variable with evidence-implied law $P_W^y$,
and let the query output lie in a common aligned or quotiented space $Y_I$. Every
$\theta\in\mathcal F_{\mathcal D}(y)$ supplies a regular conditional law $\Lambda_w^\theta$ of all private
exogenous blocks and selection seeds used by the query and a jointly measurable, well-posed post-intervention
solve $S_I^\theta$; write $K_I^\theta(w,\cdot)=(S_I^\theta)_\#\Lambda_w^\theta$. A partial-factual query is
point-identified exactly when
$\{[K_I^\theta]_{P_W^y}:\theta\in\mathcal F_{\mathcal D}(y)\}$ is a singleton, and only
$P_W^y$-almost everywhere. For a fixed finite block partition, conditional independence instead requires
$\Lambda_w^\theta=\bigotimes_b\Lambda_{w,b}^\theta$ almost everywhere. Membership in that null and calibration
of a particular test are separate questions. Loop-gain measurability is only an optional Jacobian diagnostic:
by itself it is neither necessary nor sufficient for either conclusion.
\end{definition}

\paragraph{T3: equilibrium selection.}
Within each strongly connected component, the shared selection rule is a measurable function of the private-noise
rank vector and the solution set. It is invariant under strictly increasing reparameterizations of each noise
coordinate. Thus the same ranks and the same solution set select the same branch.

\begin{definition}[Design sufficiency for a query (T4)]
Let $R_{\mathcal D}(\theta)$ be the full declared response object supplied by the experimental design, with all
labels, intervention operators, signed doses, units, and response laws retained. The design is sufficient for
the query if, for every $\theta,\theta'$ in the complete legal fibre,
\[
R_{\mathcal D}(\theta)=R_{\mathcal D}(\theta')
\quad\Longrightarrow\quad
[K_I^\theta]_{P_W^y}=[K_I^{\theta'}]_{P_W^y}.
\]
For full factuals the same definition uses the corresponding point-mass factual law. This is an observable
response condition, not an assumption that the mechanisms have already been identified.
\end{definition}

The empirical evidence has three components. (C-a) compares predicted and observed responses under staged
interventions. Agreement of means and free-block covariances is a moment condition; agreement of the complete
interventional kernels is a distributional condition. The former supports the latter only for the specific
query classes described in Theorem~\ref{thm:sound}. Even in a linear model, matched first and second moments do
not identify the noise distribution. (C-b) examines invariance of structural residuals across regimes, using
declared separators and disciplined conditional tests. (C-c) checks the structural class itself, including
monotonicity, noise assignment, and the selection rule.

For a candidate twin $\mathcal T$, write $\mathsf{Accept}_{\mathcal D}(\mathcal T)$ for the event that every
predeclared, nonempty block procedure associated with design $\mathcal D$ accepts. This notation records only
the conjunction of those procedures; it does not assert that their implications identify an arbitrary query.

\subsection{Domain changes and alignment}
A source $\mathcal S^\pi$ and target $\mathcal S^\tau$ are ECGs on a shared
graph $G$ over $V$ ($|V|=d$), each uniquely solvable per SCC under a shared $\Sel$. Define the
\textbf{discrepancy set} by
\[
\Dset=\{i:(f_i,P_{U_i})\text{ differ across domains}\}.
\]
For $i\notin\Dset$, both the mechanism and its noise law are invariant.
\emph{Exogenous independence (standing):} in \emph{both} domains the $U_i$ are
mutually independent (diagonal $\Omega$), so all cross-domain difference is localized to the per-node
$\{(f_i,P_{U_i})\}$; this excludes a shared-marginal copula shift (an invariant $\Dset$---even
$\Dset=\varnothing$---with a distinct exogenous joint law, which would reproduce every node-marginal yet move the
equilibrium law). The \textbf{cyclic selection diagram} is $D=G+\{\Sw_i\to i:i\in\Dset\}$ with
\emph{mechanism-level} semantics (Definition~\ref{def:diag}). \textbf{Linear/LQ instance:}
$V=B^\omega V+c^\omega+U^\omega$, zero-diagonal $B^\omega$, $\rho(B^\omega)<1$; rows $i\notin\Dset$ share
$(B_{i,\cdot},c_i,\Omega_{ii})$; sensor map $X^\omega=H^\omega V^\omega$, $H^\omega$ full column rank.
\textbf{Alignment $\mathcal A$:} latent frames are matched either by shared sensors ($H^\tau=H^\pi=:H$,
``anchoring'') or by the invariance-labeled frame; Proposition~\ref{prop:align} quantifies the residual
ambiguity when neither suffices.

\begin{definition}[Cyclic selection diagram; mechanism-level switches]\label{def:diag}
$\Sw_i$ annotates the \emph{mechanism} of $i$, never a solution map. The augmented model is an ioSCM whose
per-regime unique solvability holds because each fixed \emph{pure} $\Sw$-value reproduces one domain's model
(under the exogenous-independence assumption above, so all cross-domain difference is $\Sw$-annotated); for
\emph{mixed} $\Sw$-profiles (some coordinates $\pi$, others $\tau$) we carry the \citet{ForreMooij2019calculus}
SCC/loop-solvability hypothesis on the augmented model --- or, w.l.o.g.\ for Rule-D, take $\Sw$ as a
\emph{single global switch node}, which admits only the two pure profiles. The generalized directed global
Markov property applies to it.
\end{definition}

\section{Counterfactual agreement under monotone mechanisms}\label{sec:lemma}
The next lemma connects interventional distributions to unit-level counterfactuals. It provides the common
argument for validation within a domain (Theorem~\ref{thm:sound}) and transport across domains
(Theorem~\ref{thm:cf}).
\begin{lemma}[Quantile abduction under feedback]\label{lem:qa}
Let $\mathcal S,\mathcal S'\in\Tp$ share the graph, $\Sel$, the intervention semantics of $I$, and the
query-relevant \emph{interventional} kernels
\[
P\!\big(V_i\mid \operatorname{do}(V_{\mathrm{pa}(i)}{=}v_{\mathrm{pa}(i)})\big)
\]
for every $i$ on the post-surgery ancestral SCC-set of the query (with $\Sel$ rank-measurable per (T3)). Then
every full-factual per-unit counterfactual coincides: $\chi(\mathcal S;v,I)=\chi(\mathcal S';v,I)$ for a.e.\ $v$
(w.r.t.\ the factual law).
\end{lemma}
\emph{Proof.} By (T1) each mechanism is $f_i(v,u)=Q_i^{\operatorname{do}}(v,F_i(u))$, where $Q_i^{\operatorname{do}}$
is the conditional quantile of the \emph{interventional} kernel $P(V_i\mid \operatorname{do}(V_{\mathrm{pa}(i)}))$
--- the structural response with the parents held fixed --- and \emph{not} the observational
$P(V_i\mid V_{\mathrm{pa}(i)})$, which under feedback confounds $U_i$ with its own realized parents and is
therefore not the abduction object. Abduction fixes the rank $\tau_i=F_{V_i\mid \operatorname{do}(\mathrm{pa})}(v_i)$;
the post-surgery rank-coupled system is re-solved by the shared $\Sel$, which by the strengthened (T3) is
rank-measurable and so reads only the shared ranks, not the per-node noise gauge. Every ingredient is a
functional of shared \emph{interventional} objects. \hfill$\square$ The argument is the equilibrium instance of
BGM monotone-determinism \citep{NasrEsfahany2023bijective}, which is stated for structural/interventional
mechanisms rather than cyclic observational conditionals.

\paragraph{Interpretation.}
Strict monotonicity lets the factual observation recover a rank for each noise coordinate. Once the relevant
interventional kernels and the equilibrium selection rule agree, the same ranks generate the same
post-intervention outcome. Observational conditionals alone do not supply this information in a feedback system.

\noindent\textbf{Sharpness (observational conditionals do not suffice).} Two systems in $\Tp$ sharing
\emph{every} observational conditional but differing on a query-relevant interventional kernel can disagree on
$\chi$: the linear two-node witness $V_1{=}aV_2{+}U_1,\ V_2{=}bV_1{+}U_2$ with common
$\Sigma{=}\left(\begin{smallmatrix}1&0.5\\0.5&1\end{smallmatrix}\right)$ realises $\operatorname{do}(V_2)$-slopes
$a{=}0.2$ and $a{=}0.4$, while observation-only abduction returns the wrong common slope $0.5$; hence the
hypothesis cannot be weakened to observational conditionals. When only observational laws and invariance
(C-a,C-b \emph{without} staged interventions) are available, $\chi$ is set-identified with the sharp interval of
\S\ref{sec:bounds}, not a point.

\section{Validation within one domain}\label{sec:head}
We now state the conditions under which experimental agreement determines the counterfactual of interest. The
result separates the population identification requirement from the finite-sample procedures used to compare a
fixed candidate twin with the reference system.
\begin{theorem}[When experimental agreement identifies the query]\label{thm:sound}
If $\mathcal S,\mathcal T\in\Tp(\mathcal D)$ (verified by C-c) with shared graph/$\Sel$ and the staged set meets
a \emph{sufficiency} design of Theorem~\ref{thm:vmin} for $\chi$ --- one achieving $\Delta_\chi(\mathcal D)=0$
(e.g.\ any (T4) design, or the conditional-affine $V_{\min}^{\rm qry}(\chi)$ block cover; \emph{not} the bare
first-order $V_{\min}^{\rm FO}$, which the $y{+}x^2$ counterexample shows insufficient off the affine case;
\S\ref{sec:vmin}), and $\mathcal T$ passes \emph{distributional} (C-a) and (C-b)
population-exactly, then $\chi(\mathcal T;v,I)=\chi(\mathcal S;v,I)$ a.e.\ for full factuals.  For a
partial factual, equality instead requires the exact (T2) boundary: the complete legal fibre has a singleton
$P_W^y$-a.e.\ class of query pushforward kernels $[K_I^\theta]_{P_W^y}$.  Identification of the full
$(\Lambda^\theta,S_I^\theta)$ pair is sufficient but not necessary; product factorization supplies the declared
conditional independence, not equality of partial-factual query posteriors. The moment (C-a) back-tests deliver \emph{distributional} (C-a) only \emph{query-by-query}, never
blanketly: even in the \emph{linear} identified class, matching means and covariances does \emph{not} identify
the noise \emph{distribution} ($B{=}0,H{=}1$: $U\!\sim\!N(0,1)$ and $U\!\sim\!\mathrm{Exp}(1){-}1$ share every
additive-shift moment yet $P(U{>}2){=}0.023\neq0.050$). Under identified linear structure the transfer is
therefore query-specific: \emph{full-factual state-functional} queries transfer from structure and the full
factual state alone (\emph{no} noise-distribution identification); \emph{interventional moment} queries transfer
when they factor through the identified first two moments; \emph{interventional distributional} queries transfer
only under a declared \emph{moment-determined} noise family (or the distributional (C-b)/PIT route); and
\emph{partial-factual} queries only when the complete-fibre pushforward class is a singleton as in (T2). A general
\emph{nonlinear} mechanism fails even the interventional-moment step (moment-matched twins differ on the
interventional kernel and on $\chi$). So distributional (C-a) is an added, query-specific
hypothesis, \emph{not} delivered by moment matching alone.
\end{theorem}
\emph{Proof.} See Appendix~\ref{app:validation-proofs}. \hfill$\square$

\begin{proposition}[Selected-summary modulus]\label{prop:selected-modulus}
For finite samples, keep the population object and the finite-sample summary distinct. Let
$\mathcal L_{\mathcal D}(\theta)$ be the complete indexed observational and
staged laws used by Theorem~\ref{thm:vmin}.  Separately, let $\mathsf M_{\mathcal D}^{\rm mc}(\theta)$ be the
fixed, predeclared mean--covariance summary actually tested: for each of the $E$ environments it has one
reconstructed-mean block and one Frobenius-isometric $\mathrm{svec}_{\rm free}(\Sigma_{FF}^{(e)})$ block. Zero-dimensional
blocks with $F_e=\varnothing$ are omitted; let $\mathcal B_{\mathcal D}$ index the remaining blocks and put
$B=|\mathcal B_{\mathcal D}|\le2E$. After fixed,
data-independent block scalings (identity here), write
\[
 \|R_{\mathcal D}^{\rm mc}(\theta)\|_\oplus^2
 :=\|\mathsf M_{\mathcal D}^{\rm mc}(\theta)-\mathsf M_{\mathcal D}^{\rm mc}(\theta_0)\|_\oplus^2
 =\sum_{e=1}^E\bigl(\|r_e^m(\theta)\|_2^2+\|r_e^\Sigma(\theta)\|_2^2\bigr).
\]
For a fixed, validation-data-independent twin in a deterministic localization $K_{\eta,R}$, assume the
\emph{selected-summary modulus}
\[
 \|\chi(\theta)-\chi(\theta_0)\|_{\mathsf Q}
 \le C_L\|R_{\mathcal D}^{\rm mc}(\theta)\|_\oplus^s
 \qquad\forall\theta\in K_{\eta,R}, \tag{$\star_{\rm mc}$}
\]
where $\|\cdot\|_{\mathsf Q}$ is the declared query norm and $C_L,s>0$ are truth-specific instance constants.
This is a moment-to-query sufficiency and separation premise; it is \emph{not} implied by constancy on the
complete-law fibre.  A data-selected twin needs a uniform power statement for its selection rule, and a twin
outside $K_{\eta,R}$ contributes $\Pr[\mathrm{localization\ fails}]$ unless another modulus is supplied.

The modulus may be derived by {\L}ojasiewicz only when $K_{\eta,R}/{\sim}$ is a compact legal quotient that is
closed and Hausdorff, and is definable in a common \textbf{polynomially bounded} o-minimal expansion. Moreover,
$\mathsf M_{\mathcal D}^{\rm mc}$,
$\chi$, and the declared norms descend to continuous definable maps on the entire quotient (or on a finite
closed-branch decomposition with common exponent $s=\min_i s_i$ and an adjusted maximum constant); and the separate
truth-centred moment-fibre condition
\[
 \mathsf M_{\mathcal D}^{\rm mc}(\theta)=\mathsf M_{\mathcal D}^{\rm mc}(\theta_0)
 \Longrightarrow \chi(\theta)=\chi(\theta_0) \qquad(\theta\in K_{\eta,R}) \tag{Suff-M}
\]
holds.  These hypotheses give existential, not numerically computed, instance constants by the definable
{\L}ojasiewicz inequality \citep{vandenDriesMiller1996}.  Merely o-minimal is insufficient: by Miller's
dichotomy \citep{Miller1994exp}, $\chi(t)=t$ and $M(t)=e^{-1/t}$ on $[0,1]$ defeat every power modulus.
The linear/LQ class is semialgebraic (with restricted-analytic tail queries) on
$K_{\eta,R}=\{\|B\|_F\le R,\ \eta I\preceq\Omega\preceq RI,\ \sigma_{\min}(H)\ge\eta,\ \|H\|\le R,
\ \|c\|\le R,\ \eta\le c_\Sigma\le R,\ \phi\in\Xi,\ |\lambda_i(B^{(e)})|\le1-\eta\ \forall e\in
\mathcal D\cup\mathrm{patterns}(\chi)\}$, where $\Xi$ is closed, bounded and definable.  The per-pattern
spectral margins keep every relevant resolvent finite; a $\rho(B)$-only truncation does not.  The constants can
diverge at an unprobed-pattern pole, covariance floor or rank boundary.  The exponent $s$ is the exponent of
$(\chi,R_{\mathcal D}^{\rm mc})$; it equals a ratio of vanishing orders only after a uniform arc-wise proof.
The equal-moment Gaussian/centred-exponential nonlinear pair violates (Suff-M), so this moment route is
unavailable; complete-law separation cannot be substituted while retaining a moment-based rate.
\end{proposition}
\emph{Proof.} See Appendix~\ref{app:validation-proofs}. \hfill$\square$

\begin{corollary}[Finite-sample query tolerance]\label{cor:query-tolerance}
Fix the candidate independently of the validation sample, and fix all block scales before seeing that sample.
Use the Euclidean norm for reconstructed means and the Frobenius-isometric $\mathrm{svec}$ norm for covariance
blocks; their direct-sum norm is the norm in Proposition~\ref{prop:selected-modulus}. For each actual block
$b\in\mathcal B_{\mathcal D}$ assume the block-specific power condition
\[
 (\mathrm{BP}_b):\qquad \|r_b\|_b>q_{b,n}
 \quad\Longrightarrow\quad
 \Pr\{\text{block $b$ accepts}\}\le\beta_b.
\]
The aggregate decision accepts only when every block accepts. Put
\[
 Q_n:=\Bigl(\sum_{b\in\mathcal B_{\mathcal D}}q_{b,n}^2\Bigr)^{1/2},\qquad
 \varepsilon_n:=C_LQ_n^s. \tag{$\dagger$}
\]
For $B\ge1$, $\varepsilon_n\le C_L[\sqrt{B}\max_bq_{b,n}]^s$.  If $B=0$, use the empty-sum
convention $Q_n=\varepsilon_n=0$ and omit the maximum display; (Suff-M) then permits no query-changing
candidate in the localization, so the wrong-query event below is empty. For $B\ge1$ and a fixed predeclared localized twin,
$\|\chi_{\mathcal T}-\chi_{\mathcal S}\|_{\mathsf Q}>\varepsilon_n$ forces some block
$\|r_{b^\ast}\|_2>q_{b^\ast,n}$, and therefore
\[
 \begin{aligned}
 &\Pr[\mathsf{Accept}_{\mathcal D}(\mathcal T)\ \wedge\
 \|\chi_{\mathcal T}-\chi_{\mathcal S}\|_{\mathsf Q}>\varepsilon_n]\\
 &\qquad\le\beta_{b^\ast}+\Pr[\mathrm{localization\ fails}]
 \le\max_b\beta_b+\Pr[\mathrm{localization\ fails}].
 \end{aligned}
\]
This is single-test containment, not a union bound; test levels control false rejection of the correct twin.
For the nonasymptotic sub-Gaussian route with known $H$, Proposition~\ref{prop:mean-concentration} gives the mean
block $q_{e,n}^m=t_{e,n}^m(\alpha_e)+t_{e,n}^m(\beta_e)$; with estimated $H$, use
$\widetilde q_{e,n}^m=\widetilde t_{e,n}^m(\alpha_e)+\widetilde t_{e,n}^m(\beta_e)$ instead.  The covariance block order is
$q_{e,n}^{\Sigma}=O\{\bar\kappa_{\Sigma,e}[\sqrt{k_eu_e/n_e}+\sqrt{k_e}\,u_e/n_e]\}$,
$u_e=\log[4k_e/(\alpha_e\wedge\beta_e)]$, with
$\bar\kappa_{\Sigma,e}=C_BK_{Z,e}^2\bar\lambda_e$ and
$\bar\lambda_e\ge\lambda_{\max}(\Sigma_{FF}^{(e)})$, as in Proposition~\ref{prop:covariance-concentration}.
Thus both Bernstein branches and the
$\sqrt{2k_e}$ conversion from the raw-entry threshold to the $\mathrm{svec}$ norm are retained. The bound is existential/order-level:
neither $(C_L,s)$ nor a numerical $\varepsilon_n$ is estimated in the numerical illustrations, which do not
establish ($\star_{\rm mc}$) or the near-threshold theorem.

The fixed-$d$ ADF-Wald route is separate: it gives pointwise asymptotic power for a fixed nonzero residual in
the declared estimable range.  It enters no finite-$n$ sum in ($\dagger$).  A singular branch additionally needs
fixed-rank/eigengap control and the proved kernel companion (including deterministic coordinates), or it
refuses; an estimated random Wald seminorm does not replace the deterministic norm above.  The
displayed sub-Gaussian orders use independent probe observations $n_e$; clustered data require their own
cluster-sum concentration theorem, rather than a raw-row or generic effective-size substitution.  This aggregate
does not bound the $d{\times}d$ parameter error.  The identification
backbone below is the \emph{linear/LQ} class (together with the nonlinear \emph{source-block} positive result
of the companion representation paper, \citet{dadgostari_ECG1_2026}); general diffeomorphic mixing inherits that paper's
retraction and is not claimed here.
\end{corollary}
\emph{Proof.} See Appendix~\ref{app:validation-proofs}. \hfill$\square$

\paragraph{Interpretation.}
Validation is query-specific. Full interventional distributions support the general counterfactual conclusion,
whereas mean and covariance agreement is enough only when the query factors through those moments or when the
noise family is determined by them. With partial factual information, the entire legal fibre must induce the
same query distribution.

\begin{theorem}[Why the boundary conditions are needed]\label{thm:nec}
The four conditions exclude different failures.

\textbf{Without T1,} a \emph{parent-dependent} non-monotone rank-band reflection
$T_i(pa,u)=F_i^{-1}\!\circ\rho(pa,\cdot)\circ F_i(u)$ at a node whose band-driving parent is not a descendant
of $i$ is measure-preserving for every $pa$, so \emph{all} laws in every regime coincide while the
counterfactual changes on the positive-measure symmetric difference of the factual and counterfactual bands.
In the symmetric-noise linear specialization this is $u\mapsto-u$ on the band and the gap is
$2\,|e_Y^\top(I-B^I)^{-1}e_i|\,|u_i|>0$; without symmetry the gap remains positive but has no such closed form.
(A \emph{parent-independent} annulus reflection is instead a measure-preserving involution ---
an SCM isomorphism with counterfactual gap $0$ --- so the parent dependence is essential to the T1 argument.)

\textbf{Without T2,} the failure is already present without a cycle. Let
$V_1=U_1$, $V_2=U_2$, and $V_3=V_1+V_2+U_3$ with independent standard Gaussian noises, and condition on
$W=V_3$. Every loop-gain condition is vacuous, yet
$\operatorname{Cov}(U_1,U_2\mid V_3)=-1/3$, so the tensor-product identity fails. In the explicit Gaussian
difference test, calibrating $D=U_1-U_2$ with the false product variance $4/3$ instead of its true conditional
variance $2$ makes nominal level $0.05$ equal $0.109531$. Conversely, whenever the declared regular conditional
law equals the tensor product of its block marginals, the conditional independence follows regardless of loop
gain; level validity of any named test still requires its own sampling and calibration theorem. The
two-node feedback example gives a second analytic control: its conditional log-density has mixed partial
$-1/(1-v_0v_1)^2\neq0$. The same calculation has an acyclic-collider negative control and an independent-noise
positive control. Thus T2 is exactly the singleton query-kernel condition, whereas its tensor-product clause is
exactly the declared conditional-independence null; full conditional-law identification is only sufficient.

\textbf{Without T3,} take the primitive $Z\sim N(0,1)$, $f_0(Z)=\tanh Z=:u$, state space $[-2,2]$, and
best-response update
\[
 \Phi(v,u)=v-0.1\{v^3-3v-u\}.
\]
Its equilibria solve $v^3-3v-u=0$. The two outer roots are stable and are separated by more than $17/5$ for
every $u\in(-1,1)$. Two models with the same mechanism and noise, one using the minimum-stable-root rule and
the other the maximum-stable-root rule, therefore differ only in $\Sel$ and disagree by more than $17/5$ on
the corresponding query. Each model uses its declared rule in every regime. This witness is selection-specific;
it does not claim that a design which directly observes the selected branch is uninformative.

\textbf{Without T4,} let
$B_0=\left(\begin{smallmatrix}0&3/5\\1/2&0\end{smallmatrix}\right)$,
$\Omega_0=I$, $H_0=I$, $A_0=(I-B_0)^{-1}$, and let $Q_\theta$ be the planar rotation at
$\theta=0.7$. Put $A_\theta=A_0Q_\theta$, $W_\theta=A_\theta^{-1}$,
$D_\theta=\operatorname{diag}((W_\theta)_{11},(W_\theta)_{22})$,
$B_\theta=I-D_\theta^{-1}W_\theta$, $\Omega_\theta=D_\theta^{-2}$, and $H_\theta=I$.
Both models are stable Gaussian ECGs and have the same passive law because
$A_\theta D_\theta\Omega_\theta D_\theta A_\theta^\top=A_0A_0^\top$.
For factual state $(1,1/2)$ and the query obtained by setting $V_0=2$ and reading $V_1$, their values are
$1$ and $1.3916649596$. No probing-stage equality is asserted: the example shows precisely why passive
Gaussian evidence alone does not satisfy T4.
\end{theorem}
\emph{Proof.} See Appendix~\ref{app:validation-proofs}. \hfill$\square$

\paragraph{Interpretation.}
The assumptions rule out distinct failure modes rather than repeating the same restriction. Monotonicity fixes
the noise ranks, the conditional condition controls partial-factual abduction, the selection condition controls
which equilibrium is chosen, and the design condition removes query-relevant observational equivalences.

\section{Transport across domains}\label{sec:transport}
Write $\An_{G_I}(\cdot)$ for the post-surgery ancestor set. The engine is that ancestral sub-models are
autonomous.
\begin{lemma}[SCC-closed ancestral autonomy]\label{lem:auto}
For any surgery $I$ and target set $Y\cup Z$, $A:=\An_{G_I}(Y\cup Z)$ is closed under parents and under
SCC-membership, the restricted SCM on $A$ is autonomous and inherits unique solvability per SCC, and the
post-surgery law of $(V_Y,V_Z)$ is a functional of $\{f_j^\omega,P_{U_j}^\omega:j\in A\}$ alone. In the linear
class, every entry $e_w^\top(I-B_I)^{-1}$ is a walk-sum over ancestors of $w$; the identity persists by rational
continuation to all $\det(I-B_I)\neq0$ (the trek/$\perp_t$ argument). \hfill$\square$
\end{lemma}
The lemma isolates the part of the system that can affect the query after intervention. Mechanisms outside this
ancestral subsystem do not enter the post-intervention law.
\begin{theorem}[Direct transport under ancestral separation]\label{thm:direct}
Assume source identification, alignment secured on $A=\An_{G_I}(Y\cup Z)$, and the cross-domain invariance
condition (the (C-b) comparison, run with source$\,\cup\,$target regimes and conditional discipline) holds on the
mechanisms of $A$. If
\[\Sep:\qquad \An_{G_I}(Y\cup Z)\cap\Dset=\varnothing\quad(\text{no }\Sw\text{-node is an ancestor of }Y\cup Z\text{ in }D_I),\]
then $P^\tau(Y\mid\doop(I),Z)=P^\pi(Y\mid\doop(I),Z)$; in the linear class this is an all-parameter identity
(Lemma~\ref{lem:auto}).
\end{theorem}
\emph{Proof.} By \Sep{} and Lemma~\ref{lem:auto} the target law of $(Y,Z)$ under $\doop(I)$ is a functional of
$A$-indexed mechanisms, all invariant by \Sep{}$\,+\,$the cross-domain invariance condition; the same functional at the same
mechanisms is the source law, computable from $[\mathcal S^\pi]$ since identification makes the class a point on
$A$ (or $\chi$ is gauge-invariant). \hfill$\square$

\paragraph{Interpretation.}
Direct reuse requires four ingredients: the source mechanisms must be identified, the relevant latent frames
must be aligned, the retained mechanisms must satisfy the cross-domain invariance condition, and \Sep{} must
exclude every changed mechanism from the post-intervention ancestral subsystem of the query. These conditions
are sufficient; failure of \Sep{} does not by itself show that reuse is impossible for a particular query.
\begin{lemma}[Domain changes within a feedback component]\label{lem:contam}
If $i\in\Dset$ lies in an SCC $C$ with $|C|\ge2$, then every $j\in C$ is graphically downstream of $\Sw_i$.
A generic distributional conclusion needs an additional regularity premise. On any connected finite-dimensional
real-analytic stratum, suppose a finite distribution-separating signature $s_C(\theta)$ is analytic and there is
an admissible witness at which the domain change moves $s_C$. Then the parameter values for which the law of
$C$ does not change lie in a proper analytic zero set. In the linear model, changing intercept $c_i$ by $t$
moves $E(V_j)$ by $t[(I-B)^{-1}]_{ji}$ whenever that entry is nonzero, which supplies such a witness.
For an unrestricted nonlinear stratum no analytic-generic claim is made; the corresponding conclusion requires
an explicit distributional-faithfulness assumption. In every case, nodes in $C\setminus\Dset$ retain their
mechanisms even when their equilibrium laws change.
\end{lemma}
\emph{Proof.} See Appendix~\ref{app:transport-proofs}. \hfill$\square$
\begin{theorem}[Transport by replacing changed mechanisms]\label{thm:hybrid}
Fix the surgery, labelled alignment, mutually independent private noises, SCC-local selection rule, and
well-posedness conditions. Form a hybrid by taking a complete pair $(f_i,P_{U_i})$ either from the aligned source
or from target re-identification.

\textbf{Full-law transport.} If the hybrid pair equals the target pair for every nonintervened node, then the
hybrid and target have the same post-surgery joint law.

\textbf{Query-law transport.} For a query on $Y$, it is enough that those complete pairs agree on the
post-surgery ancestral SCC closure $\An_{G_I}(Y)$. Then the hybrid and target have the same law of $Y$.
Structural row recovery by itself identifies only the corresponding linear coefficients and intercepts; it
does not identify an unrestricted noise law.
\end{theorem}
\emph{Proof.} See Appendix~\ref{app:transport-proofs}. \hfill$\square$
\begin{corollary}[Moment-class hybrid transport]
If only the target structural rows and first two noise moments are re-identified, the hybrid conclusion holds
for queries that factor through those moments. A full distributional conclusion additionally requires a
moment-determined noise family or direct identification of the target noise laws on the relevant closure.
\end{corollary}
\paragraph{Interpretation.}
When a changed mechanism lies on a relevant feedback path, source reuse can be repaired by estimating that
mechanism in the target while retaining invariant mechanisms from the source. The resulting hybrid is evaluated
as one equilibrium system rather than as a patchwork of acyclic effects.
\begin{theorem}[A graphical transport rule]\label{thm:ruleD}
Treating $\Sw$ as exogenous roots, \textbf{Rule-D:} if $\Sw\sep Y\mid Z\cup X$ in $D_I$ then
$P^\pi(y\mid\doop(I),z,x)=P^\tau(y\mid\doop(I),z,x)$ for the source separator law. To use this kernel at target,
assume the regime-specific domination
$\mu_\tau^{ZX}\ll\mu_\pi^{ZX}$ for the post-surgery law of the full separator. The target integral then uses
the common kernel with the \emph{target} separator law. The conditioning set of the conclusion must contain the
full separating set $Z\cup X$, not $Z$ alone. Dropping $X$ is unsound: in $\Sw\!\to\!X\!\to\!Y$ we have
$\Sw\sep Y\mid X$ yet $P^\pi(y\mid\doop(I),z)\neq P^\tau(y\mid\doop(I),z)$ once $\Sw$ shifts $X$'s law; to
transport the $X$-marginalized $P(y\mid\doop(I),z)$ one needs $X$ intervened, the target law
$P^\tau(x\mid\doop(I),z)$, equality of the source and target $X\mid Z$ laws, or an $x$-constant kernel.
Composed with the \citet{ForreMooij2019calculus} rules on $D_I$ this licenses selection-diagram
transport formulas. Soundness holds under that paper's SCC/loop solvability hypothesis (carried verbatim); the
linear route (Lemma~\ref{lem:auto}) needs no such hypothesis. Completeness is open. \emph{Proof:} the augmented
model's generalized directed global Markov property gives $Y\indep\Sw\mid Z,X$; conditioning on the two $\Sw$
point masses gives equality $\mu_\pi^{ZX}$-a.e.; domination transports that equality to the target separator law.
\hfill$\square$
\end{theorem}
\paragraph{Interpretation.}
The graphical rule transports a conditional distribution only with the complete separating set retained. If a
separator is subsequently marginalized, its distribution must itself be transported or fixed by intervention.
\begin{theorem}[Transporting unit-level counterfactuals]\label{thm:cf}
For a per-unit counterfactual at target factuals, suppose the target hybrid lies in $\Tp$: it has monotone noise,
the complete-fibre pushforward-kernel singleton of (T2), and shared $\Sel$. Then a partial-factual $\chi$ transports
$P_W^y$-a.e.\ because every member of the complete target fibre induces the same query-pushforward kernel
by (T2); for full-factual queries, Lemma~\ref{lem:qa} applies to the hybrid, whose query-relevant \emph{interventional} kernels
are pinned by Theorem~\ref{thm:direct} or the query-law clause of Theorem~\ref{thm:hybrid} through equality of the complete
mechanism--noise pairs on the post-surgery ancestral SCC closure. Below the boundary only bounded transfer holds,
with $W=\operatorname{diam}\chi(\mathcal A)$ and, \emph{under
recombination closure and a shared $\Sel$} (else the selection-branch gap of Theorem~\ref{thm:wcw} is
added), the cross-world width bound $W\le\bar W^{\mathrm{comp}}+\bar W^{\mathrm{mono}}+
\bar W^{\mathrm{rot}}$ (Theorem~\ref{thm:wcw}; additively exact only in the decoupled
orbit-invariant-coefficient subclass; without closure only the joint diameter) at
\emph{target} parameters. \hfill$\square$
\emph{(For partial-factual abduction and transport, the query-pushforward class over the complete fibre is the
relevant object; loop gain alone does not determine it.)}
\end{theorem}
\paragraph{Interpretation.}
Transport of an interventional distribution and transport of a unit-level counterfactual are not the same
claim. The latter additionally requires the target factual evidence to select a unique query distribution over
the complete legal fibre; otherwise the appropriate answer is a range of possible values.
\begin{definition}[Complete transport fibre]\label{def:transport-fibre}
Let $\Theta$ be a complete admissible joint source--target class. It includes every profiled alignment and
nuisance parameter, all mechanisms, the joint noise law, and the full equilibrium-selection kernel for every
evidence and query regime. The underlying kernel spaces are standard Borel, and the common query-output space
has already been quotiented by the declared alignment and gauge equivalence. Fix a pre-query design $d$, a
reference $\theta_\star\in\Theta$, and the exact evidence map $E_d$, comprising all permitted source and target
design laws, including the observed target factual law but excluding the query-bearing law and unavailable
experiments. Define
\[
 \mathcal F_d(\theta_\star)=\{\theta\in\Theta:E_d(\theta)=E_d(\theta_\star)\}.
\]
Let $\mu_\star$ be the fixed factual-input law contained in $E_d(\theta_\star)$ and hence common to this fibre.
For query surgery $I$, put $Q_I(\theta)=[K_I^\theta]_{\mu_\star}$, or use the corresponding unconditional law,
in a proper quotient equipped with a separating metric $\varrho$.
\end{definition}

\begin{proposition}[Pointwise transport identification]\label{prop:pointwise-transport}
The query is pointwise transport-identified at $(d,\theta_\star)$ if and only if
$Q_I(\mathcal F_d(\theta_\star))$ is a singleton. It fails to be pointwise transport-identified if and only if
the actual fibre contains $\theta_0,\theta_1$ with
$\varrho\{Q_I(\theta_0),Q_I(\theta_1)\}>0$.
\end{proposition}
\emph{Proof.} This is the definition of identification on the complete evidence fibre; no global measurable
factorization through $E_d$ is required. \hfill$\square$

\begin{lemma}[Exact-evidence curves]\label{lem:exact-curve}
Let $\gamma:(-\varepsilon,\varepsilon)\to\Theta$, $\varepsilon>0$, satisfy
$\gamma(0)=\theta_\star$ and $E_d\{\gamma(t)\}=E_d(\theta_\star)$ throughout its domain. If some admissible
$t\ne0$ satisfies
\[
\varrho[Q_I\{\gamma(t)\},Q_I(\theta_\star)]>0,
\]
then the query is not pointwise
transport-identified.
\end{lemma}
\emph{Proof.} Both endpoints belong to the same actual evidence fibre and are separated by the query metric;
apply Proposition~\ref{prop:pointwise-transport}. \hfill$\square$

\begin{theorem}[Query-active complete-block collision]\label{thm:fail}
Let $\varepsilon>0$, let $\mathcal B$ be a complete parameter block, and consider
$\theta(t)=\iota_{\mathcal B}(c_{\mathcal B}^\star+tv;\theta_{-\mathcal B}^\star)$ for
$|t|<\varepsilon$, where $\iota_{\mathcal B}$ replaces the indicated block and leaves its complement fixed.
Assume: (i) every $\theta(t)$ belongs to $\Theta$ for every evidence and query regime,
including cross-block noise compatibility, positive definiteness, and the full selection specification;
(ii) $E_d\{\theta(t)\}=E_d(\theta_\star)$ exactly; and (iii) a quotient-well-defined scalar functional $\ell$
that is insensitive to null representatives---for example, integration against a bounded test
function---satisfies
\[
 \ell\{Q_I(\theta(t))\}=\alpha+\beta t,\qquad \beta\ne0.
\]
Then every admissible nonzero $t\in(-\varepsilon,\varepsilon)$ gives an actual-fibre collision separated by the
query, and the query is not pointwise transport-identified.
\end{theorem}
\emph{Proof.} Conditions (i)--(ii) put $\theta(t)$ and $\theta_\star$ in the same complete admissible fibre.
Condition (iii) gives distinct quotient query laws for every admissible nonzero $t$. Proposition~\ref{prop:pointwise-transport}
then gives the conclusion. \hfill$\square$

\begin{corollary}[Instantiation by the existing collision constructions]\label{cor:collision-templates}
The stable complete-block family in Theorem~\ref{thm:necessity} instantiates Theorem~\ref{thm:fail} whenever its
primitive construction establishes membership in the complete joint class for a nondegenerate interval, exact
equality of the complete design evidence, and nonzero query motion. Rotational or split-frame families, a
finite-budget family, a non-monotone reshuffle, or a T2/T4 ambiguity is an instantiation only when its own
primitive argument supplies the same three facts together with quotient-well-defined query separation. In
particular, a reshuffle outside the declared class does not prove an in-class collision, and a graphical label
or a positive-dimensional stabilizer does not by itself establish query motion.
\end{corollary}
\emph{Proof.} Each qualifying construction supplies the curve and separating functional required by
Theorem~\ref{thm:fail}. \hfill$\square$

\paragraph{Interpretation.}
Non-identification requires a design-null direction that is also query-active. Fixed, cancelled, or
surgery-excised queries fail the condition $\beta\ne0$. The intersection
$\An_{G_I}(Y)\cap\Dset$ locates a possible direct post-surgery response channel, but it is not a general
necessity condition for conditional counterfactual or selection-mediated failure. Neither raw ancestry nor
membership outside $\Tp$ is, by itself, a sufficient non-transportability claim.

\section{Identifying changed mechanisms in linear models}\label{sec:count}
We now specialize to linear models and ask how many target interventions are needed to identify the changed
mechanisms. Assume that invariance and alignment have established the rows indexed by $\Dset^c$. We consider
single-target interventions. In the $C_{FF}$-nonsingular ambient branch, active probing of all discrepancy rows identifies
the target. Smaller support-aware designs and every singular branch are governed by the complete legal fibre,
not by a universal numerical fallback.
\begin{definition}[Complete stratified linear fibre]
For a fixed design, let the complete legal fibre be the union of all admissible support, rank, normalization,
alignment, stability, and selection strata. On an affine stratum $s$, profile the fixed nuisances and write its
relative-interior patch as
\[
 \mathcal F_s^{\circ}=\{\theta_s+Q_st:M_st=0\}\cap\Theta_s^{\circ},
 \qquad Q_s\in\mathbb R^{n_s\times p_s},\quad \operatorname{rank}(Q_s)=p_s.
\]
Here $M_s$ has $p_s$ columns. Its local dimension is $p_s-\operatorname{rank}(M_s)$. Local identification is constancy on this patch; global
identification is constancy over the union of every admissible stratum and every disconnected component.
\end{definition}
\begin{lemma}[Free-block ambiguity at a singularity]\label{lem:cff-fibre}
Order the coordinates as $F,\Dset$.  Fix an invertible $C$ with unit diagonal, its known row block
$C_{F,\cdot}$, and the pinned response rows $Z=P_{\Dset,\cdot}$ of $P=C^{-1}$, and write
$Z=[Z_F\ Z_{\Dset}]$.  If
\[
q=\operatorname{nullity}(C_{FF})=\operatorname{nullity}(Z_{\Dset})
\]
and $N\in\mathbb R^{|\Dset|\times q}$ spans $\ker Z_{\Dset}$, then every compatible discrepancy-row
completion $Y=C_{\Dset,\cdot}$ is uniquely
\[
Y=Y_0+NT,\qquad T\in\mathbb R^{q\times d}.
\]
For a fixed support branch whose discrepancy-row equalities are affine, let $\mathcal L_G(Y)=b_G$ collect the
unit-diagonal equations and every declared exact zero in those rows, and put
\[
r_G=\operatorname{rank}\{T\mapsto\mathcal L_G(NT)\}.
\]
The legal completion fibre is an affine set of dimension $qd-r_G$, intersected with the open conditions that
$C$ is invertible, $\rho(I-C)<1$, all staged systems are well posed, and declared present edges remain nonzero.
Any additional nonlinear model equality requires its own constant-rank analysis.  Thus every
relative-interior patch has dimension $qd-r_G$.  In the ambient unknown-support class, where only the diagonal
is imposed, $r_G$ is the number of nonzero rows of $N$ and is generically $|\Dset|$; fixed support can increase
$r_G$ and can identify a completion even when $C_{FF}$ is singular.  Local query identification requires
constancy on the local patch; global query identification requires constancy on the entire legal fibre, including
all of its components and admissible branches.

\emph{Proof.}  The pinned-row identity $ZC=E_{\Dset}$ is
$Z_{\Dset}Y=E_{\Dset}-Z_FC_{F,\cdot}$.  Its complete solution set is $Y_0+NT$.  The displayed structural
equalities are linear in $Y$, so rank--nullity gives $qd-r_G$; intersecting strict inequalities leaves the local
dimension unchanged.  The equality of the two nullities is the complementary nullity theorem for the invertible
pair $(C,P)$.  For diagonal-only constraints, different discrepancy nodes select different columns of $T$, so
the rank is exactly the number of nonzero rows of $N$. \hfill$\square$
\end{lemma}
\begin{theorem}[Intervention requirements in linear models]\label{thm:suff}
Consider active single-target probes on distinct $\Dset$-nodes. Intervening on every discrepancy node pins the
labelled discrepancy rows $P^\tau_{\Dset,\cdot}$ of $P^\tau=(I-B^\tau)^{-1}$. Together with the invariant rows,
these probes identify $B^\tau=I-(P^\tau)^{-1}$ only when the remaining free-row completion is unique. In the
ambient unknown-support branch, that completion is unique exactly when
\[
C_{FF}P_F=E_F-C_{F\Dset}P_\Dset,\qquad C=I-B^\tau,\quad
\det(C_{FF})\ne0,
\]
or, equivalently, $\det(P^\tau_{\Dset\Dset})\ne0$. Nonsingularity of $I-B^\tau$ alone does not imply this
condition. On a singular or support-restricted branch, identification is governed instead by the complete
stratified fibre above. A scalar query is locally identified iff it is constant on every relevant
relative-interior patch and globally identified iff it is constant across all strata and components. No
numerical intervention count is inferred from singularity or local dimension alone.
\end{theorem}
\emph{Proof.} See Appendix~\ref{app:linear-proofs}. \hfill$\square$

\begin{proposition}[Joint-influence free-block inference]\label{prop:jiw-free-block}
Fix $F,\Dset$, the dimensions, environments, labels, support/rank, model branch, and $\eta_F>0$ independently
of the inference sample. A declared calibrated source experiment $\pi$ supplies the complete structural row
block $B^\pi_{F,\cdot}$ \emph{before} the target free-row solve, and the fixed common structural frame satisfies
the explicit interface assumption
\[
B^\tau_{F,\cdot}=B^\pi_{F,\cdot}.
\]
The source design has complete labelled coverage of the observational response map and every source
single-node intervention/P-row object used to reconstruct all rows of $P^\pi$ that form
$B^\pi_{F,\cdot}$; a partial collection is outside this result.
For each labelled response-map environment $e$, with a shared baseline arm $0$, probe arms $1{:}q$, and fixed
full-row-rank design $\mathsf D$, form
\[
\widehat R_e=[\bar X_{e1}-\bar X_{e0},\ldots,\bar X_{eq}-\bar X_{e0}],\qquad
\widehat M_e=\widehat R_e\mathsf D^\top(\mathsf D\mathsf D^\top)^{-1}.
\]
The labelled structural row $j$ is reconstructed, without a target $P_F^\tau$ solve, by
\[
m_j=\widehat M_0e_j,\quad D_j=\widehat M_j-\widehat M_0,\quad
w_j^\top=-\frac{m_j^\top D_j}{m_j^\top m_j},\quad
\widehat P_{jj}=\frac1{1-w_{jj}},\quad
\widehat P_{j,:}=e_j^\top+\widehat P_{jj}w_j^\top.
\]
Stacking the source rows gives $\widehat B^\pi=I-\widehat P_\pi^{-1}$. A fixed, predeclared structural
stochastic-coordinate mask retains the estimated entries and projects the known zero diagonal and every
declared exact structural zero to zero; those deterministic coordinates never enter the Wald basis. The target
reconstruction supplies only the labelled rows $\widehat P^\tau_{\Dset,:}$. Thus the pure pre-solve estimator is
\[
\widehat C=I_F-\widehat B^\pi_{FF},\qquad
\widehat G=-\widehat B^\pi_{F\Dset},\qquad
\widehat Z=\widehat P^\tau_{\Dset,:}.
\]
Let the joint primitive vector contain every source and target arm mean, each shared baseline only once, and
every nuisance actually plugged into these maps. The sampling unit is an independent cluster: within-cluster
multi-arm and source--target dependence is retained in one influence vector. If $N$ is the number of clusters,
$N_a/N\to\pi_a\in(0,1)$ for every required fixed stratum. Assume the displayed reconstruction denominators and
$P_\pi$ are nonsingular; the joint cluster array obeys a Lindeberg CLT and finite covariance (a fixed
$2{+}\epsilon$ moment bound suffices); and the complete cluster sandwich is consistent and positive definite on
the retained stochastic coordinates. The invariant full-column-rank sensor/no-direct-effect model remains in
force. The estimator does not plug in an estimated sensor. A variant that introduces an estimated sensor would
require its influence function and all source--target and nuisance cross-covariances; that extension is not proved
here. The result is not a selective-inference statement: the branch is fixed before the inference sample, or
selected on an independent split.

For conditioning alone, use the reduced source-only primitive system and let
$u_C=\operatorname{vecoff}(C)$ contain \emph{exactly} the coordinates selected by the fixed stochastic map.
Known diagonal ones and declared structural zeros are deterministic. Let $T_C$ insert $u_C$
into $\operatorname{vec}(C)$, let $J_C$ be the reconstruction Jacobian, and set
\[
\widehat V_C=\widehat J_C\widehat\Omega_\mu\widehat J_C^\top,\qquad
\Delta_N(\alpha_C)=
\sqrt{\frac{\chi^2_{r_C,1-\alpha_C}}{N}
\lambda_{\max}(T_C\widehat V_CT_C^\top)},\quad r_C=\dim(u_C)>0.
\]
Pointwise at every fixed data-generating process satisfying these assumptions,
\[
\liminf_{N\to\infty}\Pr\{\|\widehat C-C\|_2\le\Delta_N(\alpha_C)\}\ge1-\alpha_C.
\]
With $\widehat s=\sigma_{\min}(\widehat C)$, $L_N=\widehat s-\Delta_N$, and
$U_N=\widehat s+\Delta_N$, \emph{accept} the $\eta_F$-conditioning claim iff $L_N\ge\eta_F$,
\emph{reject} it iff $U_N<\eta_F$, and \emph{refuse} otherwise. On the confidence event, acceptance gives
$\sigma_{\min}(C)\ge\eta_F$ and $\|C^{-1}\|_2\le\eta_F^{-1}$; its false-accept probability and the false-reject
probability of the stated threshold claim each have asymptotic limsup at most $\alpha_C$. These are unconditional
error statements, not coverage conditional on acceptance.

The solve bound uses a separate, single joint source--target system
$u_S=(u_C,\operatorname{vec}G,\operatorname{vec}Z)$ with
$\widehat V_S=\widehat J_S\widehat\Omega_\mu\widehat J_S^\top$, dimension $r_S$, and one
$\chi^2_{r_S,1-\alpha_S}$ ellipsoid. Each vectorization retains only its declared stochastic coordinates.
If $S_j$ selects block $j\in\{C,G,Z\}$ and $T_j$ reconstructs its matrix, set
\[
\delta_j=\sqrt{\frac{\chi^2_{r_S,1-\alpha_S}}{N}
\lambda_{\max}(T_jS_j\widehat V_SS_j^\top T_j^\top)}.
\]
Writing $Q=E_F-GZ$ and $\widehat Q=E_F-\widehat G\widehat Z$, retain the full bilinear remainder
\[
\delta_Q=\delta_G\|\widehat Z\|_2+(\|\widehat G\|_2+\delta_G)\delta_Z.
\]
The full solve uses the guard from this same joint event,
\[
L_S=\sigma_{\min}(\widehat C)-\delta_C\ge\eta_F.
\]
This is distinct from the reduced $\alpha_C$ diagnostic above. On the single joint Wald event, and only after
this full-system guard accepts,
\[
\|\widehat X-X\|_2\le
\frac{\delta_Q+\delta_C\|\widehat X\|_2}{L_S},
\qquad X=C^{-1}Q,\quad \widehat X=\widehat C^{-1}\widehat Q.
\]
Equivalently, writing the displayed right-hand side as $b_N$, the pointwise unconditional error statement is
\[
\limsup_{N\to\infty}\Pr\{\text{full accept and }\|\widehat X-X\|_2>b_N\}\le\alpha_S;
\]
the deterministic inequality is not asserted outside the joint Wald event.
Separately calibrated marginal intervals may not be spliced into this solve claim. The result requires the
source-to-target invariance interface, nonvanishing allocation, a prespecified branch, valid reconstruction
denominators, compatible dimensions, and a finite positive-definite covariance on the retained coordinates;
without them the radius is undefined. Operationally, missing source provenance or invariance, invalid allocation,
post-selection on the inference sample, a failed reconstruction denominator, incompatible dimensions,
nonfinite output, or an unusable retained-coordinate covariance yields neither a radius nor acceptance. If
$F=\varnothing$ there is no free-block guard; if $|F|=1$,
$C_{FF}=[1]$ exactly and conditioning is deterministic. If a target structural row block
$B^\tau_{F,\cdot}$ is supplied externally and is exact in the declared target coordinate frame, then
$\Delta_C=\Delta_G=0$; this conclusion is conditional on the external block and its frame. This result is pointwise, fixed-dimensional, fixed-environment, fixed-branch, and
asymptotic---not finite-sample, uniform, growing-dimensional, selective, oracle, or bootstrap inference.
Finally, the conditioning check is applied to $C_{FF}$, not $P^\tau_{\Dset\Dset}$:
$B^\tau_{FF}{=}[\begin{smallmatrix}0&M\\0&0\end{smallmatrix}]$ (remaining blocks zero) has $\rho(B^\tau){=}0$ and
$P^\tau_{\Dset\Dset}{=}I$ yet $\sigma_{\min}(C_{FF})\to0$ as $M\to\infty$.
\end{proposition}
\emph{Proof.} See Appendix~\ref{app:linear-proofs}. \hfill$\square$

\begin{corollary}[Support-aware identification criterion]\label{cor:support-aware-counts}
Let $g(B)=\det(I-B_{FF})$ on a fixed admissible affine support stratum. At an admissible zero where
$Dg\ne0$, the singular locus is locally codimension one. This statement is local and regime-specific: the zero
set may be empty (in particular, if $|F|=1$ then $C_{FF}=[1]$), and fixed support can force a different rank
geometry. On a fixed corank-$q$ affine branch, Lemma~\ref{lem:cff-fibre} gives local dimension $qd-r_G$.
The support-aware intervention minimum is therefore the smallest design for which the query is constant on the
entire complete legal fibre; full identification requires every stratum and component to be a singleton modulo
the declared equivalence. Numeric counts from the $C_{FF}$-nonsingular ambient branch, from an orthogonal orbit, or from
a known-support example apply only after those regimes have been established. They are not fallbacks for an
unresolved singular branch.
\end{corollary}
\emph{Proof.} See Appendix~\ref{app:linear-proofs}. \hfill$\square$
\paragraph{Interpretation.}
The intervention requirement depends on what is already aligned and on how the target is observed. Active
probing with an unknown sensor map generally requires every changed row, whereas a shared sensor map can save
one intervention when the free-block completion is unique. Known support can reduce the requirement further.
\begin{theorem}[A lower bound for complete discrepancy blocks]\label{thm:necessity}
Assume the discrepancy set is a complete source block of size $m\ge2$, $C_{\Dset F}=0$, and the normalized local
fibre is $\mathcal A$-locally complete for the $SO(m)$ action of Proposition~\ref{prop:align}. For each design
below, also assume the proposition's normalized equivariance and local-faithfulness condition, so equality of the
anchor laws is equivalent, locally, to fixing the signed labelled anchor frame. For every design of $m-2$
linearly independent signed labelled anchors, the local survivor is $SO(2)$ and therefore contains a
nontrivial stable curve through the truth. Every query with nonzero derivative along that curve has a genuine
local collision; an invariant query does not. Hence $m-1$ anchors are necessary and sufficient to remove the
connected rotational ambiguity. This is a theorem about that normalized complete-block fibre, not a universal
intervention count. Global model or query identification still requires exclusion of transverse and disconnected
branches.

Separately, consider the ambient unknown-support Gaussian branch with distinct indices $r,s$, zero-mean Gaussian
sources having a common covariance, identical stochastic-intervention source laws in the two members, and every
retained named stage disjoint from $\{r,s\}$. Suppose $1-C_{rs}C_{sr}\ne0$ and the small split-frame perturbation
in Appendix~\ref{app:linear-proofs} remains admissible and stable. Then two members agree on all complete retained
Gaussian laws, while the labelled held-out $s$-stage response row moves; the structural response query
$(B_t)_{sr}$ also moves off the stated exceptional locus. Fixed known $H$ generically excludes
this collision. This is a scoped collision example, not a positive identification theorem and not a complete-law
claim for arbitrary support classes.
\end{theorem}
\emph{Proof.} See Appendix~\ref{app:linear-proofs}. \hfill$\square$

\paragraph{Interpretation.}
The sufficiency count cannot be read as a universal property of $|\Dset|$ alone. Below the stated boundary,
distinct stable target systems can agree on the available responses. The particular collision depends on the
observation regime, support, and completion conditions.

\begin{remark}[Numerical evaluation of $\mathcal O_\varepsilon$]\label{rem:tolerance}
Near instability, $\|\Sigma\|$ can be large enough that an absolute tolerance on
$\|\Sigma'-\Sigma\|$ falls below the corresponding floating-point scale. Numerical comparisons must therefore
use relative error. The angle should likewise be selected by tracing the connected component of the identity,
rather than by a fixed global grid, because the admissible component can shrink as $\rho(B)$ approaches one.
\end{remark}
\begin{corollary}[Query-specific intervention savings]\label{cor:cheap}
In the ambient active-probe branch with $\det(C_{FF})\ne0$, probing all $|\Dset|$ discrepancy rows identifies
the target.
In the normalized complete source-block $SO(m)$ branch, $m-1$ independent signed labelled anchors remove the
connected rotation. Outside these proved regimes there is no universal numerical formula: known support,
singular strata, alignments, and disconnected branches are summarized by the complete-fibre minimum
$K_{\min}^{\rm ID}$. For a single query the cost refines to the complete-fibre block cover of
Section~\ref{sec:vmin}; it can be smaller when every remaining direction is query-irrelevant. A bare tangent
rank gives only the first-order lower bound. \hfill$\square$
\end{corollary}
\begin{remark}[Design-relativity --- collapse under richer target designs]\label{rem:collapse}
The count is stated for the two regimes above. Two enrichments collapse it: \textbf{($\alpha$) shared \& known
$H$ with full active probing} ($M_0^\tau$ measured on all nodes): then $M_0^\tau(I-B^\tau)=H$ is a full-rank
linear system and $B^\tau$ is identified at $\Kt=0$ without using invariance (which becomes cross-domain
\emph{testable} --- a feature); \textbf{($\beta$) shared $H$ with a per-component LiNG target:} per-component
non-Gaussian ICA pins $M^\tau$ only up to the LiNG \emph{signed-permutation equivalence class}
\citep{Lacerda2008cyclicICA} --- for cyclic models uniqueness is not generic and needs the additional LiNG
stability/support conditions --- and, \emph{given that resolution}, $B^\tau$ follows at $\Kt=0$. A rotation that
is invisible to second moments can create fourth-order dependence under non-Gaussian sources; the algebraic
recovery begins only after the mixing has been resolved. Every $\Kt$ claim therefore states its target
observation model.
\end{remark}

\section{Why acyclic transport can fail under feedback}\label{sec:semantics}
\begin{theorem}[Acyclification can give the wrong transport conclusion]\label{thm:acyc}
There is a selection diagram and query for which the acyclic (intended-DAG) Pearl--Bareinboim reading concludes
that the effect transports and computes effect $0$, while equilibrium semantics give domain-dependent, nonzero
answers. For the two-agent cycle $q_i=\theta_i-b^\omega q_j+\varepsilon_i$ with $b\in\Dset$, take
$b^\pi=0.5$ and $b^\tau=0.3$. The DAG reading
$\theta_i\to q_i$ $d$-separates $\Sw$ from the query and predicts a null/invariant effect, whereas the
equilibrium do-effect is $-b/(1-b^2)$: source $-2/3\approx-0.667$ versus target
$-30/91\approx-0.330$. Acyclifying a cyclic selection
diagram is unsound. \emph{Proof:} direct computation; the DAG omits the feedback walk carrying both the effect
and the $\Sw$-dependence (Lemma~\ref{lem:auto} with $\Dset\cap\An(Y)\neq\varnothing$). \hfill$\square$
\end{theorem}
This is why the transport theory cannot be a relabeling of the acyclic theory
\citep{PearlBareinboim2014external}: the equilibrium solution map carries the domain dependence, and mechanism-level switch
semantics (Definition~\ref{def:diag}) are the correct lift.

\paragraph{Interpretation.}
Acyclic reduction can remove the feedback walk that transmits both an intervention and a domain change. A
separation conclusion on that reduced graph therefore need not describe the equilibrium counterfactual.

\section{Limits of validation from finite experiments}\label{sec:untest}
Under the witness-admissibility conditions below, validation and transport share the same indistinguishability
obstruction. We state that obstruction once.
\begin{theorem}[No unconditional cross-world validation from a finite design]\label{thm:untest}
Fix any finite design $\mathcal D$, any query surgery $I$, and any $\mathcal S^+\in\Tp$. Suppose $\mathcal S^+$
admits a node $r$ that is \emph{witness-admissible} for $(\mathcal D,I)$, i.e.
\begin{enumerate}[label=(\alph*),leftmargin=1.6em,itemsep=0pt,topsep=1pt]
\item $\mathcal D$ never probes $r$'s mechanism, and $I$ does not intervene on $r$ (a surgery on $r$ excises the
      reshuffled mechanism, so $u_r$ never enters);
\item $r$ has a parent $\pi$ that is \emph{not} a descendant of $r$ and whose value is moved by $I$ with positive
      probability. Because $\pi$ is a non-descendant it necessarily lies in a strictly earlier component of the
      condensation, so $\pi$ is exogenous to $U_r$ even when $r$ itself sits inside a cycle; this is what keeps the
      reshuffled twin a well-posed SCM (the equilibrium system $I-B$ is untouched). If no such $\pi$ exists ---
      in particular if $r$ is a root --- then only \emph{parent-independent} reshuffles are available at $r$, and
      those are measure-preserving involutions, i.e.\ SCM \emph{isomorphisms} whose counterfactual gap is
      identically zero; the conclusion below then genuinely fails.
\item $r$ is a post-surgery ancestor of the readout with non-vanishing coefficient:
      \[
      e_Y^\top(I-B^{I})^{-1}e_r\neq0.
      \]
\end{enumerate}
Then there is an $\mathcal S^-\notin\Tp$, obtained from $\mathcal S^+$ by a \emph{parent-dependent} non-monotone
measure-preserving reshuffle $T_r(\pi,\cdot)$ at $r$ --- a reflection of a band of $U_r$ whose \emph{bounds are
driven by $\pi$}, taken in rank space, $T_r=F_r^{-1}\!\circ\rho(\pi,\cdot)\circ F_r$, so that it is exactly
measure-preserving for \emph{every} continuous strictly-increasing $F_r$ --- agreeing with $\mathcal S^+$ on
\emph{every law of every regime of $\mathcal D$}, while
\[|\chi(\mathcal S^+)-\chi(\mathcal S^-)|\;\ge\;\Delta>0\quad\text{on a positive-measure factual set,}\]
namely the symmetric difference of the factual- and counterfactual-parent bands. The bound map $\pi\mapsto
b_r(\pi)$ is ours to choose and need only \emph{separate} the factual from the counterfactual parent law on a
set of positive probability, and satisfy $b_r(\pi)>a$ there (otherwise the band is empty and $\Delta\equiv0$); any strictly monotone $b_r$ bounded below by $a$ does --- the construction uses $b_r(\pi)=a+\mathrm{softplus}(\pi)>a$. (A bound depending on $|\pi|$ alone does
\emph{not} suffice in general: a sign-flipping $I$ with $\pi^{\mathrm{cf}}=-\pi^{F}$ moves $\pi$ with
probability one yet leaves the bands coincident, and the gap collapses to zero.)
\emph{In the linear class with noise at $r$ symmetric about the origin} --- so that the rank-space reflection is
the annulus negation $u\mapsto-u$ on $a\le|u|\le b_r(\pi)$ --- the gap has the closed form
$\Delta=2\,\big|e_Y^\top(I-B^{I})^{-1}e_r\big|\,|u_r|$. Without symmetry the reflection is still exactly
measure-preserving in rank space and $\Delta>0$ persists, but it is no longer $2|\cdot||u_r|$.
Hence any test accepting the correct twin under
$\mathcal S^+$ with probability $\ge1-\alpha$ accepts the \emph{same} (now wrong-by-$\Delta$) twin under
$\mathcal S^-$: \textbf{cross-world claims are necessarily conditional on class membership.} Interventional-layer
claims remain unconditionally testable (the \citealp{ForreMooij2019calculus} baseline).
\end{theorem}
\emph{Proof.} See Appendix~\ref{app:validation-proofs}. \hfill$\square$
\paragraph{Interpretation.}
A finite experiment collection cannot by itself rule out every change in the cross-world coupling. The theorem
constructs a system outside the monotone class that matches every observed experimental law but changes the
counterfactual. Validation is therefore conditional on structural class membership.
\noindent Conditions (a)--(c) identify the query regimes in which the impossibility construction applies; they are
not vacuous restrictions but the exact statement of \emph{where} the impossibility bites. A parentless (root)
node, for instance, admits only parent-independent reshuffles --- which are SCM isomorphisms with gap $0$
(below) --- so a root is never witness-admissible.
\noindent\emph{Why parent-dependence is essential.} A parent-\emph{independent} annulus reflection
$T(u)=-u$ on $|u|\in[a,b]$ is a measure-preserving \emph{involution}: $(\mathcal S,\mathrm{id})$ and
$(\mathcal S',T)$ are then the same structural model up to a noise-gauge relabelling (an SCM isomorphism), the
cross-world reshuffle cancels, and the intrinsic counterfactual gap is identically $0$. It is the dependence of
the reshuffle on a parent that moves between the factual and counterfactual worlds which breaks cross-world rank
invariance --- and this is exactly why (T1) is required. The node must be
non-descendant-parented: inside a feedback cycle the parent seen by $T_r$ is the equilibrium value, itself a
function of $U_r$, and the twin exits the well-posed class (existence/uniqueness of the equilibrium fail on a
positive-measure set).

\begin{corollary}[Cross-domain validation is conditional]\label{cor:untest-transport}
Apply Theorem~\ref{thm:untest} to the switch-augmented model of Definition~\ref{def:diag}, with $\mathcal D$ =
any finite \emph{source} design together with any finite \emph{target} design, the reshuffle placed on the
$\Sw_r$-indexed mechanism of $r$, and $r$ \textbf{witness-admissible \emph{at target parameters}} --- conditions
(a)--(c) of Theorem~\ref{thm:untest} read with $B^{\tau,I}$ and the target factual law, so in particular
$e_Y^\top(I-B^{\tau,I})^{-1}e_r\neq0$ and the target law is non-degenerate on $r$'s bound-driving parent. Then
\emph{all} source laws and \emph{all} target-design laws coincide across the pair, while the value reported by a procedure
would \emph{transport} is wrong under $\mathcal S^-$ by a gap $\Delta>0$ on positive measure; in the
\emph{symmetric-noise} specialization this closes to $\Delta=2|e_Y^\top(I-B^{\tau,I})^{-1}e_r|\,|u_r|$, whereas
under asymmetric noise $\Delta>0$ still holds but with a different value. Hence any transport procedure sound \emph{without} target-class assumptions is trivial: cross-world
transport conclusions are necessarily conditional on target $\Tp$-membership, while interventional transport
claims remain testable. \hfill$\square$
\end{corollary}
\begin{remark}[Why the target-side clause is needed]\label{rem:targetgeneric}
Witness-admissibility must be checked at \emph{target} parameters, and this is a genuinely new requirement rather
than a restatement of the within-domain one: the discrepancy rows $\Dset$ are exactly the parameters that need
\emph{not} be generic relative to the source. If a target discrepancy sets $B^{\tau}_{\cdot r}$ so that
$e_Y^\top(I-B^{\tau,I})^{-1}e_r=0$, the transported gap is $0$ even though the source-side witness is perfectly
admissible --- the impossibility simply does not bite for that query in that target. The witness may be placed at
an unprobed discrepancy node whenever the declared design leaves one available.
\end{remark}
\noindent Lemma~\ref{lem:qa} uses monotonicity as the sufficient full-factual boundary in this paper, while
Theorem~\ref{thm:untest} shows it is minimax-necessary under that theorem's witness-admissibility conditions,
within a domain and across domains alike. This is not a necessity claim for every individual query or system.
For conditional
instruments the exact null boundary is the tensor-product identity for the declared finite block partition; for
partial-factual queries it is a singleton query-pushforward class over the complete legal fibre. Full
conditional-law/map identification is sufficient but not necessary. Loop gain
is only a local diagnostic. Together, the positive and negative results locate the boundary between
identification and non-identification.

\section{Partial identification and alignment}\label{sec:bounds}
When point identification fails, the appropriate object is the complete set of query values compatible with
the retained laws and structural restrictions. We first describe the ambiguity created by rotations,
completions, and cross-world couplings, and then examine how alignment information can reduce it.
\begin{theorem}[Identified sets under structural ambiguity]\label{thm:gauge}
For an un-instrumented full-rotational source block, let $\mathcal F_\varepsilon$ be the \emph{complete}
rotation--completion fibre consistent with the design, declared support/selection stratum, denominator guards,
and $\rho(B)\le1-\varepsilon$. Under the second-moment observation model the identified set is exactly
$\chi(\mathcal F_\varepsilon)$; its rotation and transverse-completion projections are the floors
$\Wg^{\mathrm{rot}}$ and $\Wg^{\mathrm{comp}}$. In the finite-dimensional LQ branch, if the complete fibre is
encoded by finitely many semialgebraic support/selection strata and $\chi$ is a scalar semialgebraic function on
the denominator-valid locus, then $\chi(\mathcal F_\varepsilon)$ is a finite union of points and intervals whose
endpoints may be open or closed, finite or infinite. Every included value is attained; a finite extremum is
attained exactly when the corresponding endpoint belongs to the image. A closed, bounded, endpoint-attained
identified set is guaranteed when the \emph{full} fibre is compact and $\chi$ extends continuously across all
stratum closures with every query denominator bounded away from zero. These are sufficient conditions for that
conclusion, not a necessity claim. A spectral margin alone does not guarantee finiteness or attainment: the
population pinned-response algebra admits a noncompact completion-escape direction, but that illustration is
not a complete-law noncompactness proof unless its members are shown to remain in the complete fibre.

On the exact two-node normalized-rotation branch, rationalized by $t=\tan\theta$ and restricted only by strict
stability $\rho(B)<1$ (not by a fixed positive margin $\varepsilon$), the connected component of the identity has
scalar query image equal to the open interval
\[
 \left(\frac12-\frac{5\sqrt{170}}{68},
       \frac12+\frac{5\sqrt{170}}{68}\right).
\]
The endpoints are excluded because stability is strict. Per-component non-Gaussian source assumptions plus a valid acquisition/alignment result reduce the
\emph{rotation} ambiguity to signed permutations; they do not collapse transverse completions.
\end{theorem}
\emph{Proof.} See Appendix~\ref{app:alignment-proofs}. \hfill$\square$
\paragraph{Interpretation.}
The sharp identified set is the image of the complete legal fibre, not merely the image of a local tangent or a
single group orbit. Compactness and continuity guarantee endpoint attainment; a spectral stability margin alone
does not.
\begin{proposition}[Why the selection condition matters]\label{prop:sel}
On the fixed-support rational rotation branch above, the stable query image is exactly the displayed open
interval. Removing the strict stability restriction allows the readout to diverge as the diagonal normalization
approaches a regular zero. Thus T3 is load-bearing for this branch. For a complete rotation--completion fibre,
stability alone does not imply compactness, coercivity, or a nonzero query denominator. The separate
completion-escape construction concerns the pinned-response algebra only and is not used as a complete-law
counterexample.
\end{proposition}
\emph{Proof.} See Appendix~\ref{app:alignment-proofs}. \hfill$\square$
\begin{theorem}[Decomposition of cross-world ambiguity]\label{thm:wcw}
Without T1, the scalar partial-ID width is the query diameter over the joint ambiguity orbit,
$W=\operatorname{diam}\chi(\mathcal A)$; in the \emph{additive linear scalar-noise} subclass its two cross-world
parts are (no global nonlinear direct-sum claim): (a) for a node $i$ satisfying the witness-admissibility
conditions of Theorem~\ref{thm:untest} for the two worlds under comparison --- in particular, a
non-descendant parent moves across those worlds on a set of positive probability and the legal ambiguity class
contains the required parent-indexed law-preserving rearrangements --- the per-node reshuffle part
$\Wcw^{\mathrm{mono}}$ has sharp extended essential supremum
\[
 |e_Y^\top(I-B^I)^{-1}e_i|
 \{\operatorname*{ess\,sup}U_i-\operatorname*{ess\,inf}U_i\}.
\]
For an atomless noise law, measure-preserving parent-dependent rearrangements approach this value; a pointwise
maximum need not exist. Positive-mass endpoint attainment requires compatible endpoint atoms. Randomization may
split compatible atomic mass, but it cannot create endpoint atoms while preserving atomless marginals. A
parent-independent involution is an SCM isomorphism contributing $0$. Without witness-admissibility, the
essential-range display is only an upper bound for an enlarged coupling class, not the sharp width of the actual
legal fibre; the actual contribution can be smaller and is $0$ when only parent-independent relabellings are
admissible. Under witness-admissibility this ambiguity is invisible to every design law,
removed \emph{only} by assuming T1 (for nonlinear monotone mechanisms the same role is played by a
Lipschitz-modulus / transport-diameter bound); (b) the within-block rotation part $\Wcw^{\mathrm{rot}}$ is the
cross-world reading of the \emph{same} rotation floor $\Wg^{\mathrm{rot}}$ that Theorem~\ref{thm:gauge} measures
at the identification layer --- the identical floor, not an independent one: on an un-instrumented block
$\Wcw^{\mathrm{rot}}=\Wg^{\mathrm{rot}}$, and once the design identifies the rotation ($\Wg^{\mathrm{rot}}=0$ on
$\chi$'s block) $\Wcw^{\mathrm{rot}}=0$ too. The \emph{primary} object is the joint-orbit diameter
$W=\operatorname{diam}\chi(\mathcal A)$.
\end{theorem}
\emph{Proof.} See Appendix~\ref{app:alignment-proofs}. \hfill$\square$

\begin{corollary}[A uniform upper bound under recombination]\label{cor:wcw-upper}
Under \emph{recombination closure} --- $\mathcal A$ the product hull
$\mathcal A^{\mathrm{comp}}\!\times\!\mathcal A^{\mathrm{mono}}\!\times\!\mathcal A^{\mathrm{rot}}$, with
hull\,$\cap\,\Tp$ admissible and any two points joined by single-block moves \emph{each block moved at most
once}, and under a shared $\Sel$ (else the selection-branch gap below is added) --- the base-uniform
\emph{upper bound}
\[ W\;\le\;\bar W^{\mathrm{comp}}+\bar W^{\mathrm{mono}}+\bar W^{\mathrm{rot}} \]
holds, each $\bar W^\bullet$ being $\chi$'s diameter over the $\bullet$-orbit \emph{uniformly over the other
floors' admissible orbits} (corner-path telescoping; every floor worst-cased). \emph{Without} closure only the
joint diameter is guaranteed: for a non-product $\mathcal A$ the decomposition fails (the diagonal
$\mathcal A=\{(0,0),(1,1)\}$, $\chi=c+r$ --- every one-block slice a singleton, both floors $0$, yet $W=2$), and
revisiting a block gives only $W\le\sum_\bullet k_\bullet\bar W^\bullet$ (multiplicities $k_\bullet$; staircase
$\bar W{=}1{+}1$ but $W{=}4$).
\end{corollary}
\emph{Proof.} See Appendix~\ref{app:alignment-proofs}. \hfill$\square$

\begin{corollary}[Exact additivity in the decoupled subclass]\label{cor:wcw-additive}
The \emph{charged-once} additive form
$W=\Wg^{\mathrm{comp}}+\Wcw^{\mathrm{mono}}+\Wcw^{\mathrm{rot}}$ (rotation entered \emph{once} as the
identification-layer $\Wg^{\mathrm{rot}}$, mono at the truth) holds in the \emph{decoupled} regime --- floors on
transverse coordinate blocks with a query separable across them, in particular the off-block mono node below
(orbit-invariant coefficient), or when T4 pins the rotation/frame ($\Wg^{\mathrm{comp}}=\Wcw^{\mathrm{rot}}=0$).
Truth-anchored additivity \emph{fails} when the floors co-activate: a $Q$-dependent (especially sign-changing)
readout coefficient makes the joint width exceed the truth-anchored sum --- the bilinear $\chi=q\!\cdot\!u$ is
superadditive --- which is why the general bound worst-cases every floor over the others. The rotation orbit is charged \emph{once}
within the rotation group (closure under composition); the completion floor $\Wg^{\mathrm{comp}}$ acts on the
support/frame directions (the rotation and $\GL(d)$-frame gauge constructions are transverse;
completion-vs-rotation transversality in general $d$ is an additional stated condition, and the
completion$\times$mono cross-term is bounded conservatively), and the mono reshuffle acts on an \emph{off-block}
node $i\notin S$ whose reshuffle-driving parent lies in a strictly earlier SCC, so walks $i\!\to\!Y$ never enter
the source block $S$ and its coefficient $|e_Y^\top(I-B^I)^{-1}e_i|$ is orbit-invariant. It is collapsed to
signed permutations by per-component LiNG (only). A \emph{fourth} floor, the selection-branch gap (a multiplicity of admissible selections), is charged
separately and vanishes under a shared $\Sel$ (Corollary~\ref{cor:false} and Appendix~\ref{app:alignment-proofs}
track all four).
\end{corollary}
\emph{Proof.} See Appendix~\ref{app:alignment-proofs}. \hfill$\square$
\paragraph{Interpretation.}
The width below the monotone boundary combines ambiguity in the within-world model with ambiguity in how factual
and counterfactual noises are coupled. The additive display is exact only in the stated decoupled subclass.
\begin{corollary}[Moment agreement can miss gauge ambiguity]\label{cor:false}
Any validator that is a functional of second-moment laws assigns the truth's pass/fail to every
$Q\in\mathcal O_\varepsilon$. Consequently, whenever the validator passes the truth, it also passes every legal
representative on which $\chi$ differs, thereby false-validating those representatives. Restriction to the in-class
orbit is essential: an unstable representative is not an ECG and therefore does not define a false validation.
\end{corollary}
\emph{Proof.} See Appendix~\ref{app:alignment-proofs}. \hfill$\square$
\begin{lemma}[Gauge transformation of transported queries]\label{lem:gaugecov}
A transported $\chi$ is well-defined exactly when it is constant on the complete representative fibre of all legal
source representatives, target representatives, alignments, and discrete branches matching the retained
labelled laws.  On a component that has separately been proved to be a single group orbit, this reduces to
invariance under that group action; without such a transitivity proof, an orbit image need not be the complete
ambiguity set. \hfill$\square$
\end{lemma}
\begin{proposition}[Query identification under partial alignment]\label{prop:align}
Fix the complete retained baseline evidence $y_0$.  Let $\mathfrak A_0(y_0)$ be the set of all legal source
representatives, target representatives, alignments, nuisance completions, and
discrete support, normalization, sign, permutation, and selection branches matching all retained labelled laws. For
an intervention design $D$, let $a_D$ return the complete labelled interventional laws, including the operator,
target, signed dose, units, and response law, and put
\[
 \mathfrak A_D=\{\alpha\in\mathfrak A_0(y_0):a_D(\alpha)=u_0\},\qquad
 \mathcal I_D(\chi)=\{\chi(\alpha):\alpha\in\mathfrak A_D\}.
\]
Then $\mathcal I_D(\chi)$ is the sharp population identified set and $\chi$ is point identified iff its extended
diameter is zero.  Uniform identification requires this constancy on every admissible fibre.  A positive-width
claim therefore requires two surviving legal members that move $\chi$; positive stabilizer dimension alone is
insufficient.  Sharp endpoints are guaranteed to be attained when the representative fibre is compact and
$\chi$ is continuous; other mechanisms may also yield attainment.
One may quotient by a declared reparameterization only after proving that both $a_D$ and $\chi$ descend to that
quotient; the representative-fibre formulation needs no such compatibility assumption.

The following anchor count is deliberately confined to a normalized special case. Let $S=\Dset$ be a complete
source block of size $m$, with $C_{\Dset F}=0$, and suppose the normalized local residual action is $SO(m)$.
Assume $\mathfrak A_0(y_0)$ is \emph{$\mathcal A$-locally complete}: every local legal fibre direction is generated
by that action. Let $k$ anchors be signed, labelled, linearly independent directions in shared units, collected
as $A=[v_1,\ldots,v_k]$, and suppose their response law is equivariant and locally faithful in the sense that,
for $Q$ near the identity,
\[
 a_D(T_Q\alpha)=a_D(\alpha)\quad\Longleftrightarrow\quad Q^\top A(\alpha)=A(\alpha).
\]
Generically, as witnessed by a nonzero minor of the anchor
Jacobian, the local survivor stabilizer is
\[
 SO(m-k),\qquad
 \dim SO(m-k)=\binom{m-k}{2},
\]
and the removed rank is
$\binom m2-\binom{m-k}{2}=k(2m-k-1)/2$.
Thus, for $m\ge2$, $m-1$ such anchors remove the connected $SO(m)$ ambiguity. This identifies the full local fibre only under
$\mathcal A$-local completeness. Global identification additionally requires the complete fibre to contain no
transverse completion, disconnected support, reflection, permutation, normalization, or selection branch on
which $\chi$ moves. Unsigned, zero-dose, duplicate, or unlabelled anchors do not satisfy this count.

\emph{Proof.} The identified-set statement is the image of the complete survivor level set. In the normalized
special case, local faithfulness makes equality of anchor laws equivalent to fixing the $k$-frame. Its orientation-preserving stabilizer is
$SO(m-k)$; subtraction of the two Lie dimensions gives the displayed rank. The nonzero-minor premise makes this
the generic local rank. The remaining qualifications follow because the group action describes the complete
local fibre only under $\mathcal A$-local completeness and does not enumerate other components. \hfill$\square$
\end{proposition}

\paragraph{Interpretation.}
Alignment is query-specific. Interventions identify the query exactly when every legal alignment and nuisance
completion that matches the labelled laws gives the same query value. Familiar anchor counts apply only after
the residual ambiguity has been shown to have the corresponding group action.

\section{Query-specific experimental design}\label{sec:vmin}
Both halves share a design principle: the relevant cost is a minimum over \emph{environment blocks}, and is
set by the query and the complete retained-law fiber rather than by a nominal scalar rank.
\begin{theorem}[Query-specific experiment design]\label{thm:vmin}
Fix an admissible query $\chi$, differentiable at the truth $\theta_0$, and a finite unit-cost environment
library. For $\mathcal D$ in that library, let
$\mathcal K_{\mathcal D}=\{\theta':\mathcal L_{\mathcal D}(\theta')=\mathcal L_{\mathcal D}(\theta_0)\}$ be the \emph{global complete
legal fiber}: all models in the declared class that match every retained observational and staged law, after
profiling all nuisance parameters, including disconnected and discrete branches.  Here
$\mathcal L_{\mathcal D}$ denotes those complete indexed laws, not the finite moment stack used by the tests.
For a scalar query $\|\cdot\|_{\mathsf Q}$ is absolute value; for a vector query it is a declared norm. Put
\[
 \Delta_\chi(\mathcal D)=\sup_{\theta',\theta''\in\mathcal K_{\mathcal D}}
 \|\chi(\theta')-\chi(\theta'')\|_{\mathsf Q},
 \qquad
 V_{\min}^{\rm qry}(\chi)=\min\{|\mathcal D|:\Delta_\chi(\mathcal D)=0\},
\]
If the displayed set is empty, its minimum is defined as $+\infty$.
Thus $V_{\min}^{\rm qry}(\chi)$ is the minimum query-identifying intervention-block cost, and population query
identification is \emph{exactly} complete-fibre constancy. If
$T_{\mathcal K_{\mathcal D}}$ denotes the directions realised by $C^1$ feasible curves through $\theta_0$, define
$V_{\min}^{\rm FO}$ by replacing fiber constancy with
$D\chi(\theta_0)T_{\mathcal K_{\mathcal D}}=0$. For the same environment library, baseline tests, class, and
cost convention,
\[
 V_{\min}^{\rm FO}(\chi)\le V_{\min}^{\rm qry}(\chi)\le K_{\min}^{\rm ID},
\]
where $K_{\min}^{\rm ID}$ is the minimum block count giving a singleton \emph{complete} legal fiber. The first
inequality follows by differentiating constancy along feasible curves; the second because full identification
identifies every query. Neither inequality asserts universal strictness.
\end{theorem}
\emph{Proof.} See Appendix~\ref{app:design-proofs}. \hfill$\square$

\begin{corollary}[Affine row-space criterion]\label{cor:vmin-rowspace}
For an affine query $\chi(\theta)=c+G\theta$, the boundary-safe exact criterion is
$G\,\operatorname{span}(\mathcal K_{\mathcal D}-\mathcal K_{\mathcal D})=0$. A computable row-space equivalence
requires a stronger, independently verified premise. Suppose all baseline laws, nuisances, support restrictions,
and affine class equalities have been profiled and the \emph{complete} baseline fiber is exactly
$(\theta_0+\operatorname{im}Q)\cap\Theta$, where $Q$ has full column rank and $\theta_0$ is in the relative
interior of the remaining open inequalities. Suppose further that environment $j$ adds exactly the fixed linear
block $A_jQ$, so for $\mathcal D$
\[
 M_{\mathcal D}=\operatorname{stack}_{j\in\mathcal D}(A_jQ),\qquad
 \mathcal K_{\mathcal D}=(\theta_0+Q\ker M_{\mathcal D})\cap\Theta .
\]
Then
\[
 \Delta_\chi(\mathcal D)=0
 \iff \ker M_{\mathcal D}\subseteq\ker(GQ)
 \iff \operatorname{row}(GQ)\subseteq\operatorname{row}(M_{\mathcal D}),
\]
and $V_{\min}^{\rm FO}=V_{\min}^{\rm qry}$ is the minimum number of environment \emph{blocks} satisfying
the displayed row-space condition, not the scalar rank of one matrix. More generally a vector-valued full-$B$ query with selector $G_B$
requires $G_BQ\ker M_{\mathcal D}=\{0\}$. This reduces to $Q\ker M_{\mathcal D}=\{0\}$ only when the
profiled parameter is exactly $\operatorname{vec}B$; one random scalar gradient cannot identify full $B$.

In the population pinned-response algebra, $P_0=(I-B_0)^{-1}$ and a target $j$ fixes
$z_j=e_j^\top P_0$. For $B=B_0+E$,
\[
 z_j(I-B)=e_j^\top\iff z_jE=0,
\]
so $A_j(E)=z_jE$ is an exact affine block. This corollary applies only after the complete profiled-fiber premise
above has been established: pinned-response equations alone do \emph{not} characterize the complete
observational, interventional, or distributional ECG law. The illustrative calculation below therefore uses a
declared pinned-response-only affine observation algebra, not a universal complete-law claim. In its zero-diagonal,
unknown-support three-node chain, $B_{01}$ needs two unit-cost blocks while full $B$ needs three, and exact stable
twins exist below both minima; a dense four-node control gives equality $3=3$. Hence query identification can be
strictly cheaper, but a local query need not be cheaper. Changing the support changes the count.
\end{corollary}
\emph{Proof.} See Appendix~\ref{app:design-proofs}. \hfill$\square$

\begin{remark}[Limits of first-order design criteria]\label{rem:vmin-boundaries}
The row-space shortcut does not apply at an active boundary, with estimated $z_j$, for a nonlinear query, with an unprofiled
law/class/support restriction, or a near-rank block. Off the affine case, first order is only necessary:
$\chi(\theta)=\theta^2$ at $0$ and $\chi(x,y)=y+x^2$ on $\{y=0\}$ satisfy the pointwise tangent condition but
are nonconstant. Complete-fiber constancy remains the criterion; local tangent or finite-jet searches are
refutation-only across singular or incompletely enumerated branches. A finite-tolerance moment suite accesses a
localized set $\mathcal K_\tau^{\rm mc}=\{\|\mathsf M_{\mathcal D}^{\rm mc}(\theta)-
\mathsf M_{\mathcal D}^{\rm mc}(\theta_0)\|_\oplus\le\tau\}$, not the complete-law fibre; it therefore needs the
separate selected-summary modulus (and moment-fibre sufficiency) of
Proposition~\ref{prop:selected-modulus}.
\end{remark}
\paragraph{Interpretation.}
An experiment need not identify the whole model to identify a particular query. It is sufficient that the query
be constant on the complete legal fibre remaining after the experiment. A first-order rank condition gives only
a local lower bound unless the stated affine assumptions hold.
\begin{remark}[The transport instance of the same principle]\label{rem:designrank}
Theorem~\ref{thm:necessity} and Corollary~\ref{cor:cheap} are the cross-domain reading of
Theorem~\ref{thm:vmin}: where validation replaces full identification by the query-relative block cost
$V_{\min}^{\rm qry}(\chi)$, transport replaces it by the \emph{discrepancy}-relative cost. The active-probe
$C_{FF}$-nonsingular ambient branch of Theorem~\ref{thm:suff} is identified by probing all $|\Dset|$ rows; the normalized
complete-block $SO(m)$ branch has the separate anchor count of Proposition~\ref{prop:align}. Support-specific and
singular classes use their own complete-fibre $K_{\min}^{\rm ID}$ and need not obey either count. The
query-relative refinement $\Kt(\chi)$ composes the two. In every stated regime the
query-identifying design costs no more than full identification and can be cheaper, but need not be; the collapse regimes
(Remark~\ref{rem:collapse}) show
the count is design-relative rather than intrinsic.
\end{remark}

\section{Statistical procedures}\label{sec:instr}
The preceding results are population statements. This section gives procedures for comparing the reconstructed
mean and covariance responses of a fixed candidate twin with those of the reference system, under explicit
sampling, rank, and tail assumptions.
\begin{theorem}[Statistical checks and their limits]\label{thm:instr}
(i) The staged \emph{mean} back-test is run on the \emph{reconstructed latent} free-block mean
$\hat m_F^{(e)}=S_F\hat H^{+}\hat m_X^{(e)}\in\mathbb R^{|F_e|}$: $S_F$ selects the environment-$e$ \emph{free block}
$F_e$, so the dimension is $|F_e|=d-|\Dset_e|$ (\emph{not} $d$ --- e.g.\ $d{=}2$ with node $0$ hard-clamped gives
$F_e{=}\{1\}$, a one-dimensional statistic). It is the mean analogue of the covariance test's $Z=S_FH^{+}X$;
the \emph{observed} mean map $M_\mathcal T^{(e)}=HP$ is $p$-dimensional but the free block carries only $|F_e|$
random directions, so the observed-mean covariance has $\mathrm{rank}\le|F_e|<p$ and the raw $p$-dimensional
inverse does not exist. Write $m_{\mathcal T,F}^{(e)}:=E_{\mathcal T}[S_FH^+X^{(e)}]$ for the candidate's
\emph{full} environment mean, including the structural intercept and the declared hard/soft intervention semantics;
it reduces to a response-map times dose only in a baseline-centred zero-intercept specialization. Two mean routes
are distinct. \textbf{Fixed-dimensional ADF Wald:}
$T_e^m=n\,(\hat m_F^{(e)}-m_{\mathcal T,F}^{(e)})^\top\hat\Sigma_{m,F}^{+}
(\hat m_F^{(e)}-m_{\mathcal T,F}^{(e)})\to_d\chi^2_{r_{m,e}}$ under $H_0$, with
$r_{m,e}=\mathrm{rank}(\Sigma_{m,F}^{(e)})$ (generically $|F_e|$) and $\hat\Sigma_{m,F}$ the ADF estimate of
the per-observation covariance. For externally known $H$ this is the fixed-frame covariance. For an independent
sample-split $\hat H$, add $J_H\Sigma_{\hat H}J_H^\top$, where
$J_H=\partial m_F/\partial\mathrm{vec}(H)$. For a same-sample influence-function estimate, use the full joint
sandwich
$\Sigma_m+J_H\Sigma_HJ_H^\top+\Sigma_{mH}J_H^\top+J_H\Sigma_{Hm}$; omitting the cross terms is invalid.
This route is pointwise asymptotic; its finite-$n$ law is not replaced by $\chi^2$.
\end{theorem}
\emph{Proof.} See Appendix~\ref{app:inference-proofs}. \hfill$\square$

\begin{proposition}[Concentration for reconstructed means]\label{prop:mean-concentration}
Suppose $H$ is externally known and the probe observations are independent. Assume that the reconstructed free
response $Z_F^{(e)}$ obeys
\[
\|\langle Z_F^{(e)}-EZ_F^{(e)},u\rangle\|_{\psi_2}\le
K_{m,e}(u^\top\Sigma_{m,F}^{(e)}u)^{1/2},
\]
together with a declared envelope
$\bar\lambda_{m,e}\ge\lambda_{\max}(\Sigma_{m,F}^{(e)})$. For a universal $C_m$, put
\[
 t_{e,n}^m(\eta):=C_mK_{m,e}\bar\lambda_{m,e}^{1/2}
 \sqrt{\frac{|F_e|+\log(2/\eta)}{n_e}},\qquad
 q_{e,n}^m:=t_{e,n}^m(\alpha_e)+t_{e,n}^m(\beta_e).
\]
The test accepting when
$\|\hat m_F^{(e)}-m_{\mathcal T,F}^{(e)}\|_2\le t_{e,n}^m(\alpha_e)$
has level at most $\alpha_e$ and power at least $1-\beta_e$ whenever the population mean gap exceeds
$q_{e,n}^m$. For a sample-split estimated $H$, suppose the reconstructed mean residual has error at most
$h_{e,n}(\eta)$ except on an event of probability $\eta$. Define
$\tilde t_{e,n}^m(\eta):=t_{e,n}^m(\eta/2)+h_{e,n}(\eta/2)$. The estimated-$H$ test accepts at
$\tilde t_{e,n}^m(\alpha_e)$ and uses
$\tilde q_{e,n}^m=\tilde t_{e,n}^m(\alpha_e)+\tilde t_{e,n}^m(\beta_e)$; the union allocation preserves level
$\alpha_e$ and Type-II error $\beta_e$. Without such a
bound this finite-$n$ route is unavailable. No scale-free bound is possible
without the envelope, and no $(1-\rho(B))^{-1}$ substitution is made.
\end{proposition}
\emph{Proof.} See Appendix~\ref{app:inference-proofs}. \hfill$\square$

\begin{proposition}[Covariance Wald procedure]\label{prop:covariance-adf}
This procedure is stated only in a fixed latent frame in which $H$ is externally known exactly. A shared but unidentified $H$
does not suffice, because its common residual gauge is generally $GL(m)$ rather than $O(m)$. An estimated-$H$
extension would require the complete influence function, including all nuisance and same-sample cross terms;
that extension is not proved here.

Let $Z=S_FH^+X$, let
$s_e=\mathrm{svec}_{\mathrm{free}}(\hat\Sigma^{(e)}-\Sigma_{\mathcal T}^{(e)})$, and let
\[
 W_e=\operatorname{Cov}\{\mathrm{svec}_{\mathrm{free}}(\tilde Z\tilde Z^\top)\},
 \qquad k_e=|F_e|(|F_e|+1)/2.
\]
With independent probes, finite $(4+\delta)$ moments, an invertible free subsystem, and a consistent ADF
estimate $\widehat W_e$, if $W_e\succ0$ then
\[
 T_e^\Sigma=n s_e^\top\widehat W_e^{-1}s_e\ \Rightarrow\ \chi^2_{k_e}.
\]
This gives pointwise asymptotic level $\alpha$. Against every fixed nonzero covariance alternative its power
tends to one, and hence is eventually at least $1-\beta$; no finite uniform sample size is asserted.

If the population rank $r_e<k_e$ is externally declared and separated by a positive eigenvalue gap, the range
Wald and the hard-kernel companion of Proposition~\ref{prop:covariance-degenerate} are separate pieces: the
range statistic has limit $\chi^2_{r_e}$, while the companion detects fixed kernel moves. A naive pseudoinverse
over sample-created small eigenvalues is not this procedure. Raw $\mathrm{vech}$ and $\mathrm{svec}$ are
equivalent invertible coordinates in the full-rank case; $\mathrm{svec}$ fixes the direct-sum norm used above.
\end{proposition}
\emph{Proof.} See Appendix~\ref{app:inference-proofs}. \hfill$\square$

\begin{proposition}[Nonasymptotic covariance concentration]\label{prop:covariance-concentration}
For an actual covariance block $k_e\ge1$, let $W:=Z-EZ$ and suppose the
free-block responses are iid and obey the
\emph{centered whitened-shape} condition
$\|\langle W,x\rangle\|_{\psi_2}\le K_Z(x^\top\Sigma_{FF}^{(e)}x)^{1/2}$ for every $x$, with $K_Z\ge1$
(translation- and scale-invariant --- \emph{part of the stated assumptions}), together with a declared covariance \emph{envelope}
$\bar\lambda_e\ge\lambda_{\max}(\Sigma_{FF}^{(e)})$ (also an explicit assumption, since $\lambda_{\max}
(\Sigma_{FF}^{(e)})$ is an unknown population quantity --- alternatively a sample-split upper confidence bound with
error allocation, or the self-normalized \emph{empirical} Bernstein form; the test considered here is the
ADF Wald above, studentized by the empirical $\widehat W_e$ and hence oracle-free). For this analytic route use the
divisor-$n$ sample-mean-centred covariance
$\hat\Sigma_n=n^{-1}\sum_{\ell=1}^n(Z_\ell-\bar Z)(Z_\ell-\bar Z)^\top$. The identity
$\hat\Sigma_n-\Sigma=n^{-1}\sum_\ell(W_\ell W_\ell^\top-\Sigma)-\bar W\bar W^\top$ shows that the first term is
entrywise sub-exponential and the empirical-centering term is absorbed by the same $u/n$ branch. For
$u_\eta=\log(4k_e/\eta)$, the entrywise-Bernstein sup-norm test has threshold
$\tau_n(\eta)=C_1\bar\kappa_\Sigma(\sqrt{u_\eta/n}+u_\eta/n)$,
$\bar\kappa_\Sigma:=C_B\,K_Z^2\bar\lambda_e$ the \emph{operational} (envelope) constant, and level $\le\alpha$
for all $n$. Against a fixed, validation-data-independent candidate it has power $\ge1-\beta$ whenever
$\|\mathrm{svec}_{\mathrm{free}}(\Sigma_{\mathcal S}-\Sigma_{\mathcal T})\|_2>
q_{e,n}^\Sigma:=\sqrt{2k_e}[\tau_n(\alpha)+\tau_n(\beta)]$; in particular it suffices that
$n\gtrsim\bar\kappa_\Sigma^2\,(2k_e)\,\gamma^{-2}\log\tfrac{4k_e}{\alpha\wedge\beta}\ \vee\
\bar\kappa_\Sigma\,\sqrt{2k_e}\,\gamma^{-1}\log\tfrac{4k_e}{\alpha\wedge\beta}$ (the \emph{same} $\bar\kappa_\Sigma$
as in $\tau_n$; it equals the oracle $\kappa_\Sigma$ below when the envelope is tight, $\bar\lambda_e\!\approx\!
\lambda_{\max}(\Sigma_{FF}^{(e)})$, and otherwise inflates the guaranteed sample size by a factor in
$[\bar\lambda_e/\lambda_{\max},(\bar\lambda_e/\lambda_{\max})^2]$ --- the $\max$ of the quadratic variance-branch
term and the linear far-branch term, following whichever dominates and possibly switching branch)
--- \emph{two} Bernstein branches, the
sub-Gaussian variance branch carrying $2k_e$ and the sub-exponential far branch $\sqrt{2k_e}$ ($k_e\!\asymp\!d^2/2$)
--- with the oracle $\kappa_\Sigma=C_B\,K_Z^2\lambda_{\max}(\Sigma_{FF}^{(e)})$; when
$\Sigma_{FF}^{(e)}\succ0$, this also equals $C_B\,K_Z^2\lambda_{\min}(P_{\mathrm{free}}^{(e)})^{-1}$ for
$P_{\mathrm{free}}^{(e)}:=(\Sigma_{FF}^{(e)})^{-1}$, the \emph{actual marginal} free-block
precision (the centred-product norm $\|W_iW_j-\Sigma_{ij}\|_{\psi_1}\le C_BK_Z^2\sqrt{\Sigma_{ii}\Sigma_{jj}}\le
C_BK_Z^2\lambda_{\max}(\Sigma_{FF}^{(e)})$; the structural
$(I-B_{FF}^{(e)})^\top\Omega_{FF}^{-1}(I-B_{FF}^{(e)})$ equals it \emph{only} under a hard clamp --- a soft probe
leaves the complementary block random and inflates $\Sigma_{FF}$),
hence, on that positive-definite branch, $\kappa_\Sigma^2\asymp K_Z^4\lambda_{\min}(P_{\mathrm{free}}^{(e)})^{-2}$
(universal constant $C_B^2$). The $\lambda_{\max}$ form remains valid for singular $\Sigma$. $K_Z^4$ is a valid Bernstein \emph{upper-bound} scale (not exact/minimax-necessary), \emph{not}
$(2{+}\bar\kappa_4)$: the excess-kurtosis form is \emph{false} --- for the continuous near-Rademacher
$U=(R{+}\varepsilon G)/\sqrt{1{+}\varepsilon^2}$, $2{+}\bar\kappa_4\to0$ as $\varepsilon\to0$ while
$\mathrm{Var}(U_iU_j)=1$ stays fixed, so off-diagonal detection still needs $n=\Omega(\gamma^{-2})$. Here $K_Z$ is a
\emph{declared} assumption of the result; a diagnostic $\psi_2$ pretest can reject gross violations
but cannot establish an upper bound on $K_Z$.
\end{proposition}
\emph{Proof.} See Appendix~\ref{app:inference-proofs}. \hfill$\square$

\begin{proposition}[Rank-degenerate covariance inference]\label{prop:covariance-degenerate}
In the fixed identified frame of Proposition~\ref{prop:covariance-adf}, this branch applies exactly when the \emph{environment's}
$W_e$ is degenerate, with rank/kernel read off $W_e$ --- a \emph{population} selection rule, taken
here as an externally supplied \emph{fixed} rank/kernel declaration (part of the assumptions), \emph{not} a
data-driven finite-sample rank estimate, since exact rank is not selectable uniformly near degeneracy.
This branch lies outside the standing continuous-noise class; it records the limiting structural failure rather
than enlarging the class used by the identification theorems.
This branch inherits the \emph{full} standing hypothesis set of the fixed-dimensional case --- mutually \emph{independent} free-block noises,
each of \emph{positive variance} $\sigma_j^2{>}0$ and finite $(4{+}\delta)$ moments, \emph{zero mean}
$\mathbb E U_j{=}0$, and an \emph{invertible} free subsystem $A{=}(I{-}B_{FF})^{-1}$ (i.e.\ $\det(I{-}B_{FF})\ne0$)
--- \emph{not} merely the zero-mean convention: degeneracy of $W_e$ is a \emph{fourth-moment} statement
(on $\mathrm{Var}(U_j^2)$) layered \emph{on top of} these standing assumptions, never a consequence of centering
alone. Since $W_e$ is built
on \emph{mean-centered} variables, the exact degeneracy condition on coordinate $j$ is
$\mathrm{Var}((U_j{-}\mathbb E U_j)^2){=}0$, which under $\mathbb E U_j{=}0$ reads $\mathrm{Var}(U_j^2){=}0$.
\emph{Only} under a \emph{hard clamp} (where $Z=A_{FF}U_F$, so such a $U_j$ degenerates
$W_e$ along $Q_j$) --- or when $F$ contains every still-random coordinate --- the kernel is the
$r$-dimensional span of the \emph{congruence} pullbacks $Q_j=A^{-\top}E_{jj}A^{-1}$ of the degenerating
diagonal directions, $A{=}(I{-}B_{FF})^{-1}$ (in the fixed Frobenius-isometric coordinates the dual vectors are
$q_j=\mathrm{svec}(Q_j)$, and both projectors use the same $\mathrm{svec}$ convention), with
$r{=}\#\{j:\mathrm{Var}((U_j{-}\mathbb E U_j)^2){=}0\}$ counting the
\emph{exogenous} free-block noises of \emph{centered constant magnitude} (\emph{not} the reconstructed
coordinates $Z_i$: mixing gives $\mathrm{Var}(Z_i^2)>0$ from the $U_jU_k$ cross-term even when the underlying
noise has centered constant magnitude --- e.g.\ $Z_1{=}U_1{+}aU_2$, $Z_2{=}U_2$ has
$\mathrm{Var}(Z_1^2){=}4a^2{>}0$ yet $\mathrm{rank}\,W_e{=}1$, $r{=}2{=}m_e$), so $r{=}m_e$ iff every
exogenous free-block noise is symmetric two-point/Rademacher up to scale. Each coordinate is
\emph{independent}: a coordinate $j$ with $\mathrm{Var}((U_j{-}\mathbb E U_j)^2){>}0$ is non-degenerate and
drops \emph{only its own} $Q_j$ from the kernel, so in general
$\ker W_e=\mathrm{span}\{Q_j:\mathrm{Var}((U_j{-}\mathbb E U_j)^2){=}0\}$ and $W_e{\succ}0$
(kernel dim $0$) \emph{iff} $\mathrm{Var}((U_j{-}\mathbb E U_j)^2){>}0$ for \emph{every} $j$; one failing
coordinate leaves the other $Q_j$ intact (e.g.\ a $2$-D clamp with $U_1{\in}\{1,-2\}$, $U_2$ Rademacher gives
$W_e{=}\mathrm{diag}(2,4,0)$ in $\mathrm{svec}$ coordinates, kernel dim $1$ not $0$). Two ways a \emph{single} coordinate fails: an
\emph{asymmetric} zero-mean two-point $U{\in}\{1,-2\}$ has $\mathrm{Var}(U^2){=}2{>}0$; and a
\emph{non-zero-mean} constant-magnitude $U$ ($\mathbb P(U{=}1){=}\tfrac34$, $\mathbb P(U{=}{-}1){=}\tfrac14$:
$|U|{\equiv}1$ but $\mathbb E U{=}\tfrac12$) has $\mathrm{Var}((U{-}\mathbb E U)^2){=}\tfrac34{>}0$. Under a \emph{soft} probe $Z=V_F$
carries the complementary block's noise, and this \emph{may remove some or all} kernel directions: writing
$Z=LU$ over the independent still-random Rademacher noises, $\mathrm{Var}(\tilde Z_i^2)=4\sum_{k<l}L_{ik}^2L_{il}^2>0$
\emph{iff free coordinate $i$ genuinely mixes $\ge2$ noises} (row $i$ of $L$ has $\ge2$ nonzeros); a lone-noise
coordinate keeps its diagonal kernel direction. This per-coordinate mixing is \emph{necessary} but not sufficient
(e.g.\ $Z_1{=}U_1{+}U_2$, $Z_2{=}U_1{-}U_2$ mixes both yet $\tilde Z_1\tilde Z_2{\equiv}U_1^2{-}U_2^2{=}0$, rank $1$;
in general $\mathrm{rank}\,W_e\le p(p{-}1)/2$ for $p$ still-random noises). Full rank therefore requires the full
\emph{mixing/spanning} non-degeneracy of the congruence/moment map, not genericity alone --- e.g.\ $Z_1{=}U_1{+}aU_3$,
$Z_2{=}U_2$ leaves $\tilde Z_2^2{\equiv}1$ and $W_e$ degenerate for every $a$ (eigenvalues $\{0,2(1{+}a^2),4a^2\}$),
whereas the \emph{scalar} free block ($m_e{=}1$, $F{=}\{1\}$) $Z_1{=}U_1{+}aU_2$ mixes both and $W_e{=}4a^2\succ0$ is
full rank. The branch is selected from the actual $W_e$: when a soft probe restores
full rank the full-rank Wald applies; when it leaves a
residual kernel (declared $r$, which need not be $m_e$ and whose kernel need not be $\mathrm{span}\{Q_j\}$) the
two-piece test applies. When degenerate,
the test has \emph{two} pieces.
\textbf{(i) Range Wald:} the rank-truncated Wald (spectral cut at the declared rank $k_e{-}r$; \emph{not} the naive
Moore--Penrose inverse of $\widehat W_e$, whose $O_p(1/n)$ sample eigenvalues re-inflate each kernel direction to a
non-negligible, \emph{convention-dependent} nonstandard contribution --- e.g.\ $(1{-}X^2)^2/(4X^2)$, $X\sim N(0,1)$,
under the unbiased sample-covariance convention, \emph{not} a $\chi^2$) is $\chi^2_{k_e-r}$ on $\mathrm{range}(W_e)$,
with level $\alpha$ and power against range-supported alternatives. \textbf{(ii) Companion kernel test:} the $r$
kernel functionals equal their population values up to $O_p(1/n)$ under $H_0$ (population-exact; the residual is only
the $O_p(1/n)$ centering fluctuation) but shift by $\Theta(1)$ under a twin moving a kernel direction, so a
threshold $t_n$ with $1/n\!\ll\!t_n\!\ll\!1$ (e.g.\ $t_n\!\asymp\!n^{-1/2}$) has level $\to0$ and power $\to1$
against any $\Omega(1)$ kernel move. Purely-kernel alternatives are caught only by (ii), range alternatives by (i).
This is \emph{not} a ``sub-Gaussian threshold'' (two-point noise \emph{is} sub-Gaussian; the failure is degeneracy,
not tails, and the full-rank Wald does not apply --- the sup-norm test still applies with $K_Z\!\asymp\!1$).
\end{proposition}
\emph{Proof.} See Appendix~\ref{app:inference-proofs}. \hfill$\square$

\begin{corollary}[Aggregate query tolerance]\label{cor:aggregate-tolerance}
For the fixed validation-independent candidate, fixed block scales, direct-sum Euclidean/$\mathrm{svec}$ norm,
and block-specific conditions $(\mathrm{BP}_b)$ of Corollary~\ref{cor:query-tolerance}, the finite-$n$
sub-Gaussian tolerance is exactly the block expression in
Corollary~\ref{cor:query-tolerance},
$\varepsilon_n=C_L[\sum_{b\in\mathcal B_{\mathcal D}}q_{b,n}^2]^{s/2}$.
For $B\ge1$, the bound by $C_L[\sqrt{B}\max_bq_{b,n}]^s$ is only a conservative display; it does not replace the block sum.
Here a known-$H$ mean block uses $q_{e,n}^m$, an estimated-$H$ mean block uses $\tilde q_{e,n}^m$, and the operational covariance radius retains
both Bernstein branches and its declared-envelope $\bar\kappa_\Sigma$; the covariance rate cannot be read off a
mean constant. Against $(B,c_\Sigma\Omega)$ every mean test passes. Let $S$ be the source covariance with
unbiased divisor $n-1$ and suppose $\Sigma_\mathcal T=c_\Sigma\Sigma_\mathcal S$. Then the exact Gaussian pivot is
\[
c_\Sigma(n-1)\operatorname{tr}(\Sigma_\mathcal T^{-1}S)\sim\chi^2_{m(n-1)}.
\]
For $D=\operatorname{tr}(\Sigma_\mathcal T^{-1}S)-m$,
$E(D)=m(1-c_\Sigma)/c_\Sigma$,
$\operatorname{Var}(D)=2m/[c_\Sigma^2(n-1)]$, and the exact squared signal-to-noise ratio is
$(n-1)m(1-c_\Sigma)^2/2$. In a separate orientation check, take source covariance
$\Omega'=\mathrm{diag}(1{+}a,1{-}a,1)$ and comparator $I$. Then the trace statistic has zero mean gap and is
blind to this trace-free change, whereas the full Wald $T_e^\Sigma$ rejects it. Reversing source and comparator
is not blind because $\operatorname{tr}(\Omega'^{-1})-3=2a^2/(1-a^2)>0$. The numerical illustrations exercise both block types of
$R_{\mathcal D}^{\rm mc}$ (means and free-block $\mathrm{svec}$), but do not establish the selected-summary modulus;
$\widehat W_e^{+}$ needs $n>\mathrm{rank}\,W_e$, and noise below the fourth moment is
outside the result. Only spanned directions are validated (counterpart: a span-deficient wrong twin --- the split-gauge
completion family on the untested directions; finite level and power calculations do not address this identification failure).
\end{corollary}
\emph{Proof.} See Appendix~\ref{app:inference-proofs}. \hfill$\square$

\begin{proposition}[Invariance diagnostics and their ceiling]\label{prop:invariance-ceiling}
Invariance is \emph{necessary}; under the (T2) factorization null a named conditional test is
level-valid only with its separate sampling/calibration theorem. It also has an existential ceiling. There are
exact-law, quotient-valid pairs with invariant residual laws but different admissible queries. One is the passive
Gaussian counterfactual pair in Theorem~\ref{thm:nec}; another is the parent-dependent measure-preserving
reshuffle in Theorem~\ref{thm:untest}. These examples use query functionals that descend to the declared
equivalence class, not entries of a particular representative $H$. Thus (C-b) tests environment invariance and
can detect drift, but it does not by itself identify a cross-world query. The same ceiling applies cross-domain,
as Corollary~\ref{cor:untest-transport} shows.
\end{proposition}
\emph{Proof.} See Appendix~\ref{app:inference-proofs}. \hfill$\square$

\paragraph{Interpretation.}
The mean and covariance procedures test selected implications of a candidate twin; they do not establish the
structural class or identify an arbitrary counterfactual. Their conclusions are limited to the reconstructed
directions, sampling regimes, and rank conditions stated in the theorem.

\section{Numerical illustrations}\label{sec:exp}
\paragraph{Design and scope.}
The figures use small synthetic linear ECG systems with planted source--target discrepancies and fixed seeds.
They illustrate selected constructions and finite-instance behavior; they do not establish theorem validity,
identification, completeness, coverage, or operational performance.
\begin{figure}[tbp]\centering
\includegraphics[width=0.82\textwidth]{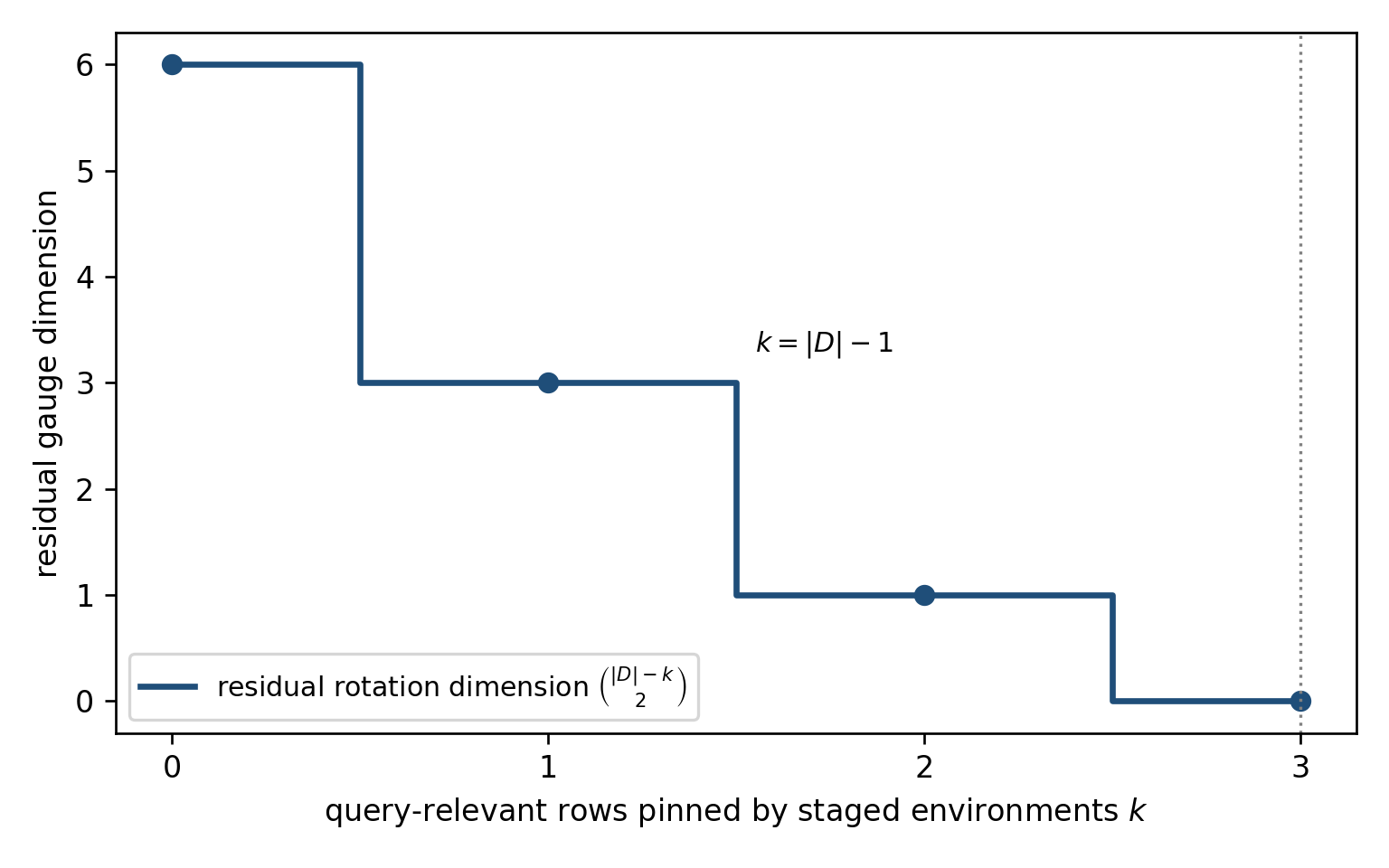}
\caption{\textbf{Interventions reduce rotational ambiguity in the displayed source block.}
Each additional query-relevant row removes rotational degrees of freedom; the residual dimension reaches zero
after $|D|-1$ rows are pinned. This calculation concerns the stated rotational orbit and does not give a
universal intervention count.
\label{fig:dash}}
\end{figure}
\begin{figure}[tbp]\centering
\includegraphics[width=0.95\textwidth]{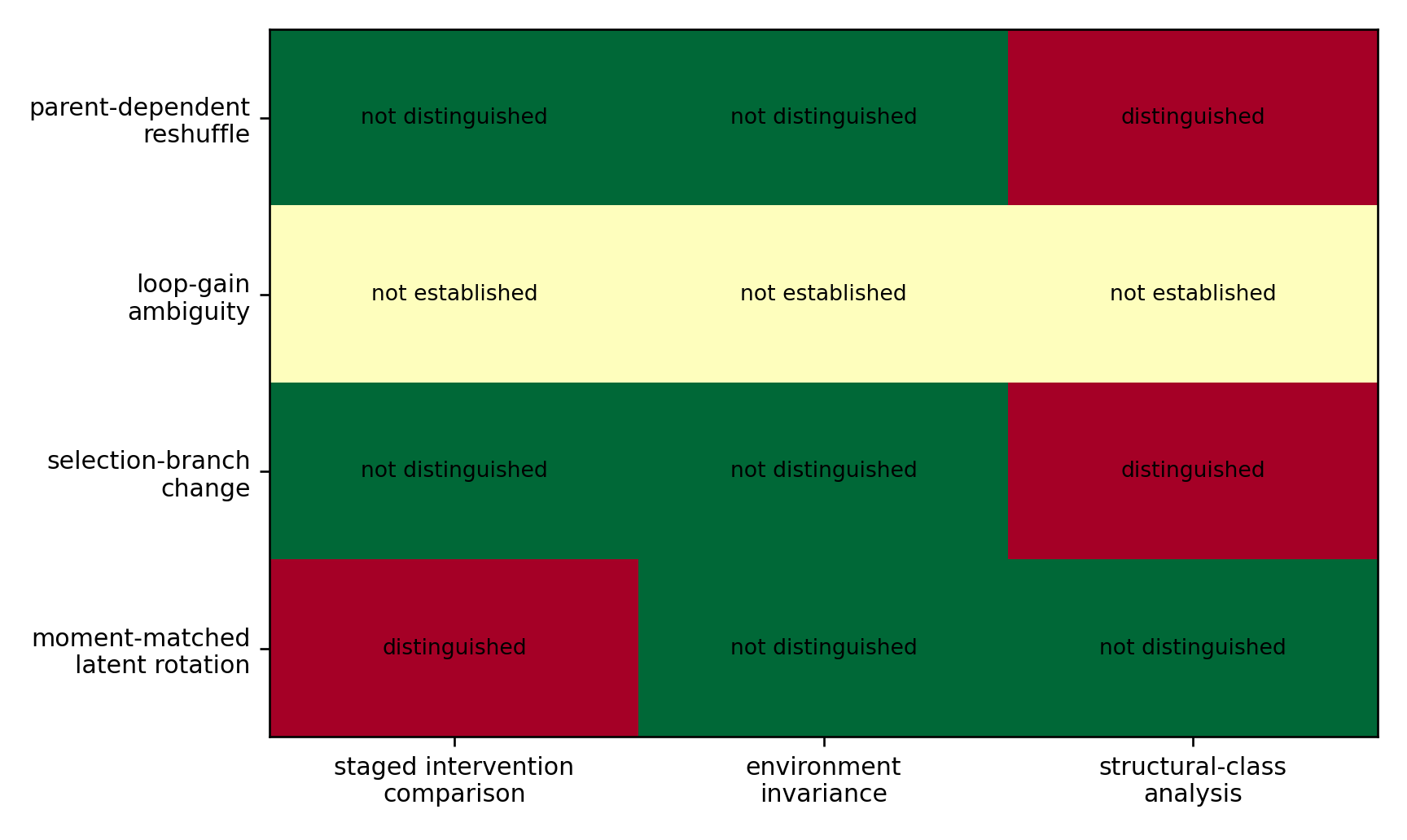}
\caption{\textbf{Different model discrepancies require different diagnostics.}
The cells summarize the displayed constructions. A staged intervention distinguishes the moment-matched latent
rotation, whereas structural-class analysis distinguishes the parent-dependent reshuffle and the
selection-branch change. The loop-gain row is marked ``not established'' because no matching construction is
claimed here.}
\label{fig:diagnostics}
\end{figure}
\begin{figure}[tbp]\centering
\includegraphics[width=0.82\textwidth]{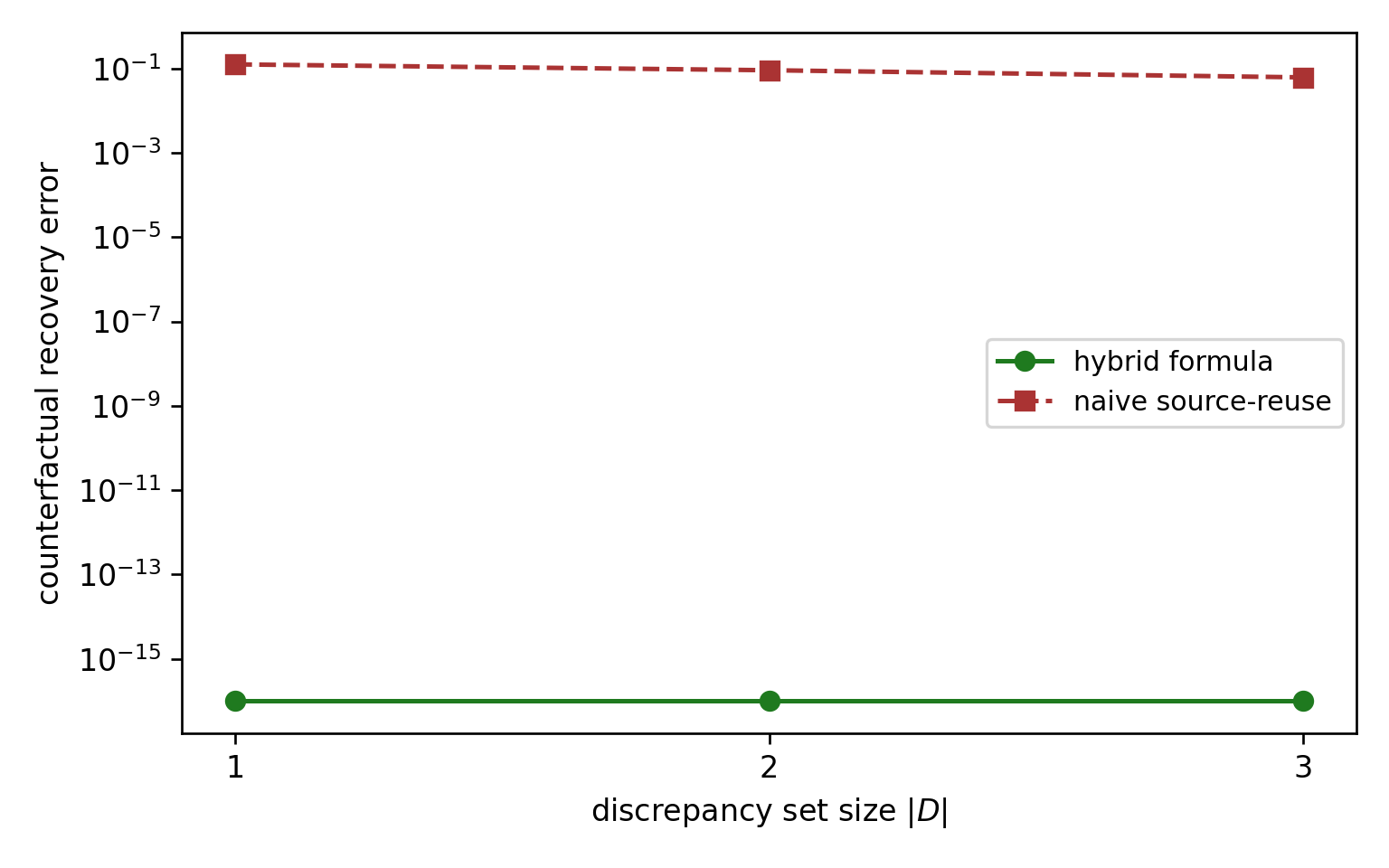}
\caption{\textbf{Hybrid reconstruction versus direct source reuse.}
In the displayed synthetic systems, the hybrid model of Theorem~\ref{thm:hybrid} reproduces the target response
to numerical precision, whereas direct reuse of the source model retains the domain discrepancy.}
\label{fig:transport}
\end{figure}
\begin{figure}[tbp]\centering
\includegraphics[width=0.96\textwidth]{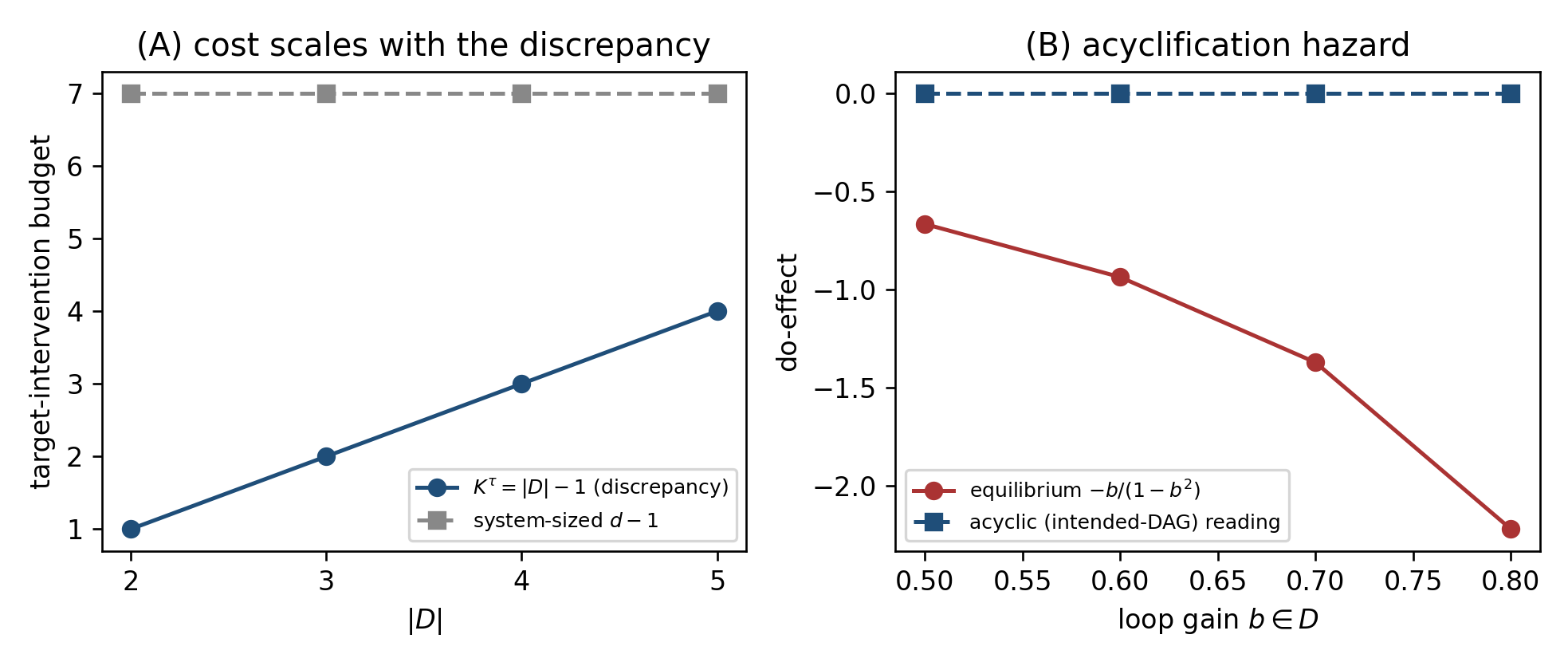}
\caption{\textbf{Design relativity and the acyclification hazard.}
\textbf{Panel A:} in these examples, the target-intervention budget scales with the discrepancy set rather than
the full system; no universal count is implied. \textbf{Panel B:} an acyclic reading misses the nonzero
equilibrium effect in the example of Theorem~\ref{thm:acyc}.}
\label{fig:dr}
\end{figure}
\paragraph{Interpretation.}
Figures~\ref{fig:dash} and~\ref{fig:diagnostics} distinguish experimental information from structural
restrictions. Figures~\ref{fig:transport} and~\ref{fig:dr} illustrate when direct reuse fails and why the
observation model and support conditions must accompany every intervention count.

\section{Discussion and limitations}\label{sec:limitations}
\paragraph{Relation to existing causal theory.}
The hedge constructions of Shpitser and Pearl concern identification in the acyclic causal hierarchy. The
finite-design obstruction here serves a related logical purpose, but it concerns the cross-world coupling and
the equilibrium response of a feedback system. Likewise, cyclic separation and intervention calculus describe
relations among interventional distributions, while a transported per-unit counterfactual additionally requires
factual abduction, equilibrium selection, and alignment across domains. These distinctions explain why neither
acyclic hedges nor cyclic separation alone settles the questions studied here.

\paragraph{Cross-world scope.}
The positive results require the declared boundary class. Monotonicity is an assumption, supported by the
constructed necessity result and the additive linear specialization. For partial factual information,
membership in a fixed-partition tensor-product null, calibration of a named test, and constancy of the query
distribution over the complete legal fibre are distinct claims; loop gain alone establishes none of them. The
completeness of the graphical transport rule remains open.

\paragraph{Design-relative identification.}
Intervention counts depend on the observation model, support, environment library, incoming-block information,
and complete legal fibre. The affine expression $qd-r_G$ applies only at fixed corank and support, within a
fixed-branch relative-interior patch. The query-targeted comparison does not imply universal strict savings,
and the alignment result does not imply a universal one-anchor rule.

\paragraph{Statistical scope.}
Finite-sample query guarantees based on moments require both the selected-summary modulus and the stated
moment-fibre sufficiency condition.
Constancy on a complete-law fibre cannot replace either condition. The joint-influence free-block Wald procedure
is pointwise and asymptotic, with fixed dimension, environment, and branch. The covariance rank-degenerate branch
requires its declared rank, eigengap, and kernel conditions. An estimated-sensor extension of the free-block
procedure is not proved here.

\paragraph{Evidence and external validity.}
The numerical material consists of synthetic or analytic illustrations. It does not establish general
identification, completeness, coverage, or performance in an operational system. Evaluation on substantive
application domains remains future work.

\paragraph{Open problems.}
Three boundaries remain unresolved. First, a complete graphical characterization of transport under feedback is
not yet available. Second, the exact intervention requirement for incomplete discrepancy blocks with live relay
mechanisms remains structure-dependent. Third, finite-sample procedures with an estimated sensor map and
data-dependent covariance rank require additional analysis. The conditional cumulant hierarchy discussed in the
supporting theory should be regarded as a direction for future work, not as a result established here.

\section{Conclusion}
For feedback-capable causal twins, validation and transport are governed by the same basic decision logic.
The mechanism-level separation criterion gives a sufficient condition for direct reuse within the boundary
class; when it fails, direct reuse may still hold for a particular query.
When separation fails but the source and target frames are aligned, the shared selection rule is retained, and
the target discrepancy mechanisms and their noise laws are exactly re-identified, hybrid reconstruction
recovers target post-surgery laws. Per-unit counterfactual transport additionally requires the target
boundary-class conditions; otherwise the complete legal fibre provides the partial-identification answer within
the declared class. This separation between direct reuse, target re-identification, and irreducible ambiguity
is the main practical consequence of the theory.

\section*{Availability}
Code and data for the numerical illustrations are available from the author upon reasonable request and will be
released publicly upon journal publication.

\appendix
\section{Technical arguments for validation}\label{app:validation-proofs}

This appendix collects the arguments used in the main text. Definitions and remarks do not require separate
proofs. When a result concerns a complete legal fibre, all nuisance, support, alignment, normalization, and
selection branches named in the statement are part of that fibre.

\subsection{Rank abduction and partial factual information}

\paragraph{Lemma~\ref{lem:qa}.}
Strict monotonicity gives
$f_i(v_{\mathrm{pa}(i)},u)=Q_i^{\doop}(v_{\mathrm{pa}(i)},F_i(u))$, where $Q_i^{\doop}$ is the conditional
quantile of the structural intervention that fixes the parents. Thus a full factual observation fixes
$\tau_i=F_i(u_i)$ almost surely. An increasing reparameterization of $u_i$ preserves $\tau_i$. The
post-surgery system may therefore be written entirely in terms of the shared ranks, the shared interventional
kernels, and the selection rule. Rank-measurability of $\Sel$ makes the two selected solutions equal. This is
why observational conditionals cannot replace the parent-fixing interventional kernels in a feedback model.

For partial factual information, standard-Borel disintegration gives a regular conditional law
$\Lambda_w^\theta$, unique $P_W^y$-almost everywhere. Its pushforward through the well-posed solve is
$K_I^\theta$. Hence all data-consistent models give the same posterior query law exactly when
$\{[K_I^\theta]_{P_W^y}\}$ is a singleton. Conditional independence of a declared finite partition is instead
equivalent to the tensor-product identity for $\Lambda_w^\theta$; it neither implies nor is implied by equality
of arbitrary query pushforwards.

\subsection{Population validation}

\paragraph{Theorem~\ref{thm:sound}.}
At a query-sufficient design, equality of the complete staged interventional kernels places the candidate and
reference systems in the same complete law fibre. Constancy of the query on that fibre identifies the
query-relevant kernels. Lemma~\ref{lem:qa} then gives the full-factual conclusion, and the singleton
pushforward condition above gives the partial-factual conclusion. The invariance comparison extends the
kernel equality from the staged regimes to every regime used by the query.

The moment-only alternatives require a separate argument. Under identified linear structure, a full-factual
state-functional query is fixed by the factual state and the structural map; an interventional moment query is
fixed when it factors through the identified first two moments. A distributional query additionally requires a
family determined by those moments or equality of the complete interventional kernels. The Gaussian and
centred-exponential example in the theorem has equal first two moments but unequal tail probabilities, proving
that this additional premise cannot be omitted.

\paragraph{Proposition~\ref{prop:selected-modulus}.}
On the compact, closed, Hausdorff definable quotient, define
$g(\theta)=\|\chi(\theta)-\chi(\theta_0)\|_{\mathsf Q}$ and
$h(\theta)=\|R_{\mathcal D}^{\rm mc}(\theta)\|_\oplus$. The condition (Suff-M) gives
$h^{-1}(0)\subseteq g^{-1}(0)$. The definable {\L}ojasiewicz inequality in a common polynomially bounded
o-minimal expansion therefore supplies $g\le C_Lh^s$ for instance-dependent $C_L,s>0$. A finite closed-branch
decomposition uses the minimum exponent and the maximum adjusted constant. The example
$g(t)=t$, $h(t)=e^{-1/t}$ shows why polynomial boundedness is needed. Per-pattern spectral margins and the
remaining compactness conditions make every relevant rational solve continuous; a margin on the observational
matrix alone does not control an unobserved surgery.

\paragraph{Corollary~\ref{cor:query-tolerance}.}
If the query error exceeds $C_LQ_n^s$, the modulus gives
$\|R_{\mathcal D}^{\rm mc}\|_{\oplus}>Q_n$. Since
$Q_n^2=\sum_bq_{b,n}^2$, at least one block has residual norm greater than its own $q_{b,n}$. Acceptance of all
blocks is then contained in acceptance by that one separated block, giving its Type-II bound. This is a
single-block containment argument, not a union bound. The empty-block convention makes the wrong-query event
empty under (Suff-M).

\subsection{Necessity constructions}

\paragraph{Theorem~\ref{thm:nec}.}
For T1, apply a parent-dependent rank-band reflection at a node whose band-driving parent is not its descendant.
For every fixed parent value the map preserves the noise law, so every designed regime has the same law. When
the intervention moves that parent, the factual and counterfactual bands have a positive-measure symmetric
difference and the readout changes. A parent-independent reflection composes with itself and is only a noise
relabeling, which explains the parent-dependence requirement.

For T2, take independent standard Gaussian $U_1,U_2,U_3$ and
$V_3=U_1+U_2+U_3$. Conditioning on $V_3$ gives
$\operatorname{Cov}(U_1,U_2\mid V_3)=-1/3$. Thus the factorization null fails even without a cycle; using the
product variance $4/3$ in place of the true variance $2$ gives the size stated in the theorem. The converse is
the defining tensor-product identity for conditional independence.

For T3, let $Z\sim N(0,1)$, let the surgery indicator be $s\in\{0,1\}$, and use the state space $[-2,2]$ with
\[
 f_0(v,Z)=\tanh Z,
 \qquad
 f_1(v,Z)=v-\frac1{10}\{v^3-3v-\tanh Z\}.
\]
At $s=0$ the equilibrium is uniquely $v=\tanh Z$.  At $s=1$, writing $a=\tanh Z\in(-1,1)$, the equilibrium
equation $v^3-3v-a=0$ has roots
$r_-(a)\in(-2,-1)$, $r_0(a)\in(-1,1)$, and $r_+(a)\in(1,2)$.
Since $\partial_vf_1=1.3-0.3v^2$, the outer roots are stable and the middle root is unstable.  The two models
use, consistently across all their regimes, either the minimum-stable or the maximum-stable SCC-local
selector.  They therefore agree in every design that does not stage $s=1$, but disagree after $\doop(s=1)$.
Indeed $r_+(a)-r_-(a)>17/5$ for all $a\in(-1,1)$: the gap is even, decreases for $a>0$, and at the endpoint
$a=1$ its square is $12-3x^2>741/64>289/25$, where the middle root is $-x$ and $0<x<3/8$.  The example uses
the paper's declared-selection notion of well posedness; it would not satisfy a requirement of raw uniqueness
of every fixed point.

For T4, put
\[
B_0=\begin{pmatrix}0&3/5\\1/2&0\end{pmatrix},\quad
C_0=I-B_0,\quad \Omega_0=I,
\]
and, for the planar rotation $Q_\theta$, define
\[
d_1=\cos\theta-\tfrac12\sin\theta,\quad
d_2=\cos\theta+\tfrac35\sin\theta,\quad
\Lambda_\theta=\operatorname{diag}(d_1^{-1},d_2^{-1}),
\]
\[
C_\theta=\Lambda_\theta Q_\theta^\top C_0,
\qquad B_\theta=I-C_\theta,
\qquad \Omega_\theta=\Lambda_\theta^2.
\]
Whenever $d_1d_2\ne0$,
$C_\theta^{-1}\Omega_\theta C_\theta^{-\top}=C_0^{-1}C_0^{-\top}$, so the passive zero-mean Gaussian law is
unchanged.  At $\theta=.7$ both models are stable, whereas for factual $v=(1,1/2)$, intervention
$\doop(V_0=2)$, and response $Y=V_1$, the counterfactual gap is
\[
B_{\theta,10}-\tfrac12
=\frac{7\tan(.7)}{2\{3\tan(.7)+5\}}>0.
\]
This collision applies only to passive stages or stages lying in the exact stabilizer of the displayed law;
a generic labelled structural probe separates the pair.  The two necessity constructions are independent.

\paragraph{Theorem~\ref{thm:untest} and Corollary~\ref{cor:untest-transport}.}
Choose the rank-space reflection so its upper band boundary separates the factual and counterfactual laws of
the non-descendant parent on a set of positive probability. It preserves every conditional mechanism law and
hence every regime in the finite design, but its two cross-world applications do not cancel on the symmetric
difference of the bands. Multiplication by the nonzero resolvent coefficient gives a positive query gap. Under
symmetric noise the reflection is $u\mapsto-u$ on the band and yields the displayed $2|c||u|$ formula; without
symmetry only positivity of the gap is retained. The cross-domain result applies the same construction to the
target switch-indexed mechanism and uses the target resolvent and target factual law.

\section{Technical arguments for transport under feedback}\label{app:transport-proofs}

\paragraph{Lemma~\ref{lem:auto}.}
In the condensation graph, the ancestor set of $Y\cup Z$ after surgery is closed under parents and under
strongly connected components. Solving components in condensation order shows that its post-surgery law uses
only mechanisms and noise laws in that set. In a linear model, expanding $(I-B_I)^{-1}$ in its convergent
walk series proves the claim under $\rho(B_I)<1$; multiplying by the determinant extends the resulting rational
identity to the entire nonsingular locus.

\paragraph{Theorems~\ref{thm:direct} and \ref{thm:hybrid}.}
Under \Sep, Lemma~\ref{lem:auto} restricts the query law to mechanisms that are invariant and aligned, so the
source and target functionals coincide.  For the hybrid result, write
$\kappa_i^\omega=(f_i^\omega,P_{U_i}^\omega)$ for the complete aligned mechanism--noise pair.  Full target-law
transport requires $\kappa_i^h=\kappa_i^\tau$ for every nonintervened row, identical surgery and clamps,
independent noises, identical variable/noise alignment, the same SCC-local selection semantics, and a
well-posed selected solve in every post-surgery SCC.  Induction along the condensation graph then gives the
same law block by block.  For a query on coordinates $Q$, equality is needed only on the post-surgery ancestral
SCC closure $\operatorname{An}^{\rm SCC}_{G_{I_e}}(Q)$, yielding equality of the $Q$-law.  Equality of only a
specified class of moments transports only queries determined by that class; it does not imply equality of the
complete law unless a separate moment-determinacy premise is supplied.

\paragraph{Lemma~\ref{lem:contam}.}
Fix a connected real-analytic support stratum.  If a single measurable functional $h$ makes
\[
\Delta(\theta)=\int h\,dP^1_{\theta,j}-\int h\,dP^0_{\theta,j}
\]
finite and analytic, and an explicit point on that stratum has $\Delta\ne0$, then equality of the two laws is
contained in the proper analytic zero set $\{\Delta=0\}$.  It consequently has empty interior and chart
Lebesgue measure zero.  The argument must be repeated on every connected stratum; where no separator and
witness are available, the conclusion requires an explicit distributional-faithfulness assumption.

For the linear-intercept specialization, replacing $c_i$ by $c_i+t$ gives
$V^{(t)}-V^{(0)}=t(I-B)^{-1}e_i$ under the common-noise coupling, and hence
\[
E[V_j^{(t)}]-E[V_j^{(0)}]=t[(I-B)^{-1}]_{ji}.
\]
The exact exceptional set is $t=0$ together with the zero set of the corresponding cofactor of $I-B$.
On a connected analytic support stratum, one nonzero cofactor witness makes this a proper analytic set.
Measure-preserving reparameterizations, or strata on which the cofactor vanishes identically, show why no
unqualified global genericity claim is valid.

\paragraph{Theorem~\ref{thm:ruleD}.}
The generalized directed global Markov property of the augmented model gives
$Y\indep\Sw\mid Z,X$ after surgery. Evaluating the conditional law at the source and target point masses of
$\Sw$ yields equality conditional on the full separator $Z\cup X$.  For this source version of the conditional
kernel to define the target integral, the regime-specific target separator law must satisfy
\[
\mu_{\tau,e}^{ZX}\ll\mu_{\pi,e}^{ZX}.
\]
Then, for the common target-a.e. version $K_{\pi,e}$,
\[
P_\tau^e(Y\in A)=\int K_{\pi,e}(A\mid z,x)\,\mu_{\tau,e}^{ZX}(dz,dx).
\]
Marginalizing $X$ requires the target conditional law of $X$ given $Z$.  Reusing the source-marginalized kernel
is valid only when that conditional law agrees in the two domains target-a.e., or when the kernel is
$x$-constant on the relevant target support.  Matched marginals alone do not suffice: source mass on $X=Z$
and target mass on $X=1-Z$ give a direct counterexample.

\paragraph{Theorem~\ref{thm:cf}.}
The hybrid fixes the query-relevant interventional kernels. Lemma~\ref{lem:qa} proves the full-factual claim,
while the complete-fibre singleton proves the partial-factual claim. Below point identification, the answer is
the image of the target legal fibre; the width statements follow from Corollaries~\ref{cor:wcw-upper} and
\ref{cor:wcw-additive}.

\paragraph{Definition~\ref{def:transport-fibre}, Proposition~\ref{prop:pointwise-transport},
Lemma~\ref{lem:exact-curve}, Theorem~\ref{thm:fail}, and Corollary~\ref{cor:collision-templates}.}
The standard-Borel premise makes every $Q_I(\theta)$ a well-defined equivalence class under the common factual
law. Identification is constancy on the actual fibre, so a separated pair is both necessary and sufficient for
pointwise failure. An exact-evidence curve lies in that fibre. In Theorem~\ref{thm:fail}, the affine scalar
functional changes by $\beta t$, so every admissible nonzero point on the declared interval is separated from
the reference. Existing constructions instantiate this argument only after establishing global membership,
exact evidence equality, the full selection/noise requirements, and query motion. This explains why ancestry,
a graph label, or stabilizer dimension alone is insufficient.

\paragraph{Theorem~\ref{thm:acyc}.}
Solving the two-node equilibrium gives the effect $-b/(1-b^2)$. Substitution of $b=.5$ and $.3$ gives
$-2/3$ and $-30/91$. The acyclic graph deletes the feedback walk and returns zero, so it cannot represent this
domain dependence.

\section{Technical arguments for linear target identification}\label{app:linear-proofs}

\paragraph{Lemma~\ref{lem:cff-fibre}.}
The pinned-row identity is
$Z_{\Dset}Y=E_{\Dset}-Z_FC_{F,\cdot}$. Its complete solution is $Y_0+NT$. Applying the affine structural
equalities to $NT$ gives rank $r_G$ and hence dimension $qd-r_G$. Strict stability, invertibility, and
present-edge inequalities do not alter the relative-interior dimension. The complementary-nullity theorem for
the inverse pair $(C,P)$ gives the equality of nullities.  The associated determinant-zero locus is locally
codimension one only at an admissible regular zero, that is, where the derivative of $\det C_{FF}$ is nonzero
along an admissible tangent direction.  A model-compatible witness is obtained from
\[
V=\begin{pmatrix}1&0\\0&1\\1&0\\0&1\end{pmatrix},\qquad
W=\begin{pmatrix}0&1\\1&0\\0&-1\\-1&0\end{pmatrix},\qquad B=VW^\top.
\]
Here $B^2=0$ and a two-by-two free block of $C=I-B$ is
$\left(\begin{smallmatrix}1&-1\\-1&1\end{smallmatrix}\right)$; its determinant has a nonzero admissible
tangent derivative.  At rank deficiency at least two the adjugate vanishes, so the determinant locus may be
singular and no global hypersurface claim follows.

\paragraph{Theorem~\ref{thm:suff}.}
Active probing pins the labelled discrepancy rows of $P=(I-B)^{-1}$. The invariant rows supply the equations
$C_{FF}P_F=E_F-C_{F\Dset}P_\Dset$. Thus $C_{FF}$ nonsingular gives the unique ambient solve, and inversion gives
$B$. On a fixed singular support branch, Lemma~\ref{lem:cff-fibre} supplies the exact remaining affine fibre.
For the normalized rotational calculation on a complete source block of size $m$, let $A$ be the matrix of
$k$ signed, labelled, linearly independent anchor directions.  Equivariance makes the anchor-preserving
stabilizer
\[
\{Q\in SO(m):Q^\top A=A\}\cong SO(m-k).
\]
Its dimension is $\binom{m-k}{2}$, so the anchors remove
$k(2m-k-1)/2$ rotational directions and $m-1$ generic anchors eliminate that continuous orbit.  The generic
rank assertion follows from a nonzero minor at the identity anchor frame and persists on an open dense set of
the fixed-support chart.  This orbit calculation becomes a statement about the complete local data fibre only
under the theorem's local-completeness hypothesis: in a neighborhood of the reference point, every legal
data-equivalent model must lie on this normalized anchor-preserving orbit.  Global identification further
requires the complete legal fibre to have no transverse, disconnected, or discrete branches.

\paragraph{Proposition~\ref{prop:jiw-free-block}.}
For each independent cluster, stack the source, target, shared-baseline, and nuisance influence contributions
before taking the outer product.  The multivariate cluster central limit theorem and delta method give one
joint affine chart for $(\widehat C,\widehat G,\widehat Z)$.  Construct a single $(1-\alpha_S)$ ellipsoid in
that chart; all three radii $\delta_C,\delta_G,\delta_Z$ are projections of this same event, not separately
calibrated marginal intervals.  Their covariance includes every overlap and plug-in cross term named in the
statement.  Weyl's inequality gives
$\sigma_{\min}(C)\ge L:=\sigma_{\min}(\widehat C)-\delta_C$, and the procedure refuses unless $L\ge\eta$.
With $Q=E_F-GZ$,
\[
\delta_Q=\delta_G\|\widehat Z\|
 +(\|\widehat G\|+\delta_G)\delta_Z.
\]
Finally,
$C(\widehat X-X)=(C-\widehat C)\widehat X+(\widehat Q-Q)$ and
$\|\widehat Q-Q\|\le\delta_Q$ give
\[
\|\widehat X-X\|\le
\frac{\delta_Q+\delta_C\|\widehat X\|}{L}.
\]
Thus the unconditional probability of accepting while the displayed bound fails has limsup at most
$\alpha_S$.  A reduced guard using its own $\alpha_C$ is a separate construction and cannot be combined with
this conclusion without a new allocation.  Every listed refusal corresponds to a missing premise. In
particular, an estimated sensor would alter the complete joint influence vector and is not supplied here.

\paragraph{Corollary~\ref{cor:support-aware-counts}, Theorem~\ref{thm:necessity}, and
Corollary~\ref{cor:cheap}.}
Unknown incoming edges leave $d-|\Dset|$ unconstrained directions in an unprobed row. With known incoming edges,
the missing-row completion reduces to the universal-parent condition on the stated no-relay block. In the
shared-sensor covariance regime the complete unpinned block has the orthogonal deficit above. For the lower
bound, every choice of $|\Dset|-2$ probes leaves two rows whose stable rotation remains inside a complete
full-rotational component and preserves the exact designed laws. Choosing a point in the identity component
keeps $\rho(B)<1$; the query-motion premise provides a separated member. In incomplete support the rotation may
leave the class, which is why no uniform lower bound is asserted there. The count comparison in
Corollary~\ref{cor:cheap} is therefore restricted to the ambient class stated in that corollary; support-aware
classes use $K_{\min}^{\rm ID}$.

For the unknown-sensor collision used in the lower-bound discussion, choose distinct indices $r,s$, put
$u=e_s-C_{rs}e_r$, and define
\[
E_t=tu e_r^\top,\quad C_t=C(I+E_t),\quad
P_t=(I+E_t)^{-1}P,\quad H_t=H(I+E_t),\quad \Omega_t=\Omega.
\]
Then $\operatorname{diag}(CE_t)=0$, $H_tP_t=HP$, and
\[
P_t=P-\frac{t}{1-tC_{rs}}(e_s-C_{rs}e_r)e_r^\top P.
\]
For a Gaussian zero-mean source with common covariance, identical stochastic-intervention sources, and every
retained named stage disjoint from $\{r,s\}$, the P1 intervention identity makes the complete retained loading
matrix—and hence every recorded joint Gaussian law—identical.  In the ambient support model the perturbation
is legal for small $t$, and
$e_s^\top P_t=e_s^\top P-\frac{t}{1-tC_{rs}}e_r^\top P$, so the labelled held-out $s$-stage response row moves.
Moreover, $\frac{d}{dt}(B_t)_{sr}|_{0}=-(1-C_{rs}C_{sr})$, so the structural response query moves away from
the exceptional locus. This is a quotient-valid query separation. Fixed known $H$ generically excludes the
collision, and a stage involving $r$ or $s$, a forbidden support entry, a nonidentical intervention source, or
a nonzero source mean not transformed consistently invalidates the construction.

\section{Technical arguments for partial identification and alignment}\label{app:alignment-proofs}

\paragraph{Theorem~\ref{thm:gauge} and Proposition~\ref{prop:sel}.}
By definition, the sharp identified set is the image of the complete legal fibre. On finitely many
semialgebraic strata, the Tarski--Seidenberg theorem makes the scalar image a finite union of points and
intervals. Compactness plus continuity guarantees a compact image and attainment of its extrema. A spectral
margin does not bound transverse completion coordinates, so it cannot provide that conclusion alone. Removing
the stability restriction lets the rotation approach the diagonal-normalization pole.  For example, with
\[
B_0=\begin{pmatrix}0&3/5\\1/2&0\end{pmatrix},\qquad t=\tan\theta,
\]
the normalized rotation gives
\[
B_{12}(t)=\frac{2(5t-3)}{5(t-2)},\quad
B_{21}(t)=\frac{5(2t+1)}{2(3t+5)},\quad
\chi(t)=\frac{13t+10}{2(3t+5)}.
\]
The stable component containing the identity is the open interval
$((1-\sqrt{170})/13,(1+\sqrt{170})/13)$, and its query image is exactly
\[
\left(\frac12-\frac{5\sqrt{170}}{68},
      \frac12+\frac{5\sqrt{170}}{68}\right).
\]
The endpoints are excluded stability-boundary models.  An exact unbounded response-row family is
$B_t=\left(\begin{smallmatrix}0&t\\1/(4t)&0\end{smallmatrix}\right)$,
$H_t=I-B_t$, $\Omega_t=I$: then $H_t(I-B_t)^{-1}=I$, $\rho(B_t)=1/2$, and
$\det(I-B_t)=3/4$, while the labelled response $B_{12}=t$ is unbounded.  This latter calculation is a
response-row illustration, not a claim that every staged law is equal.

\paragraph{Theorem~\ref{thm:wcw}.}
The query diameter is taken over the joint ambiguity set, so it is the sharp primary object. A
parent-dependent measure-preserving reshuffle attains, in the extended-essential-supremum sense, the
scalar-noise width multiplied by the resolvent coefficient when the node is witness-admissible for the two
worlds and the legal ambiguity class contains the corresponding parent-indexed rearrangements. More precisely,
under those conditions, if $U$ is atomless with essential endpoints $a,b$ and $c$ is the relevant coefficient,
then
\[
\sup_{T_0,T_1}\operatorname*{ess\,sup}
 |c|\,|T_1(U)-T_0(U)|=|c|(b-a),
\]
where the supremum is over law-preserving rearrangements and may be infinite.  The antitone quantile coupling
gives the lower bound.  This is not generally a pointwise maximum: for $U\sim\mathrm{Unif}(0,1)$ the width is
one but the endpoint gap is not attained almost surely. Positive-mass endpoint attainment needs compatible
endpoint atoms. Randomization may split existing compatible atomic mass, but cannot create endpoint atoms while
preserving atomless marginals. A parent-independent
involution is a relabeling and contributes zero. Without witness-admissibility, the same essential-range
quantity is only an enlarged-coupling upper bound; it need not equal the actual legal-fibre width. The rotation
term is the same orbit already present in the
within-world identified set, so it is not counted twice.

\paragraph{Corollary~\ref{cor:wcw-upper}.}
Under recombination closure, join any two admissible points by a path that changes each ambiguity block once.
The triangle inequality bounds each step by the corresponding diameter uniformly over the other blocks.
Summing gives the upper bound. The diagonal and staircase examples show respectively why product closure and
the one-change path are required.

\paragraph{Corollary~\ref{cor:wcw-additive}.}
In the decoupled subclass the query is a sum of functions on transverse coordinate blocks, so the supremum and
infimum separate and the diameter is the sum of the three diameters. A shared selection rule removes the fourth,
branch-selection term. A bilinear readout violates this separability and explains why the general result is only
the uniform upper bound.

\paragraph{Corollary~\ref{cor:false} and Lemma~\ref{lem:gaugecov}.}
A second-moment functional is constant on every legal second-moment orbit. If the query is not constant there,
the same decision is assigned to two different query values. More generally, a transported query descends to
the retained object exactly when it is constant on the entire representative fibre; a group orbit suffices
only after transitivity has been established.

\paragraph{Proposition~\ref{prop:align}.}
The first claim follows by retaining every legal member of the labelled-law level set. For a complete source
block of size $m$, $k$ signed, labelled, independent anchor directions form a matrix $A$ satisfying the
equivariance relation $A(T_Q\theta)=Q^\top A(\theta)$. Assume additionally the local faithfulness condition
$a_D(T_Q\theta)=a_D(\theta)$ if and only if $Q^\top A(\theta)=A(\theta)$ for $Q$ near the identity. Hence the residual normalized rotational stabilizer is
\[
\operatorname{Stab}_{SO(m)}(A)\cong SO(m-k),
\]
with dimension $\binom{m-k}{2}$.  Thus $m-1$ generic anchors eliminate the normalized continuous $SO(m)$
orbit.  A special query can require fewer anchors precisely when this residual stabilizer acts trivially on
that query.  These are local orbit statements.  Only the explicitly assumed local-completeness condition
identifies the orbit with the complete nearby data fibre; global point identification also requires the absence
of transverse, discrete, and disconnected branches.  Without equivariance the survivor is a general level set,
and tangent rank alone cannot remove those branches.

\section{Technical arguments for query design and inference}\label{app:design-proofs}

\paragraph{Theorem~\ref{thm:vmin}.}
Query identification is constancy on $\mathcal K_{\mathcal D}$, so minimizing over the environment library gives
$V_{\min}^{\rm qry}$. Constancy differentiates to zero along every feasible curve, proving
$V_{\min}^{\rm FO}\le V_{\min}^{\rm qry}$. A singleton complete fibre identifies every query, proving the
second inequality.  When the legal fibre has finitely many relative-interior strata and, on stratum $s$, the
reference-law equality is exactly $\theta_0+Q_sZ_s$, where $Q_s\in\mathbb R^{n_s\times p_s}$ has full column rank,
while the additional environments impose exactly $M_sz=0$ with $M_s$ having $p_s$ columns, the complete fibre is
\[
\bigcup_s\{\theta_0+Q_sz:z\in Z_s\cap\ker M_s\}.
\]
At a relative-interior point its local dimension is
$p_s-\operatorname{rank}M_s$.  A query is globally identified if and only if it is constant on every connected
component of every stratum and those constants agree across all components and strata.  For an affine query
$L\theta$, within-component constancy is $LQ_s\ker M_s=\{0\}$; the cross-component comparisons remain
necessary.  Thus zero local dimension or full Jacobian rank on one branch does not by itself establish global
identification.

\paragraph{Corollary~\ref{cor:vmin-rowspace}.}
Under the complete affine-fibre premise, the surviving displacements are exactly $Q\ker M_{\mathcal D}$.
Constancy of $c+G\theta$ is therefore $GQ\ker M_{\mathcal D}=0$. The finite-dimensional fundamental theorem of
linear algebra makes this equivalent to
$\operatorname{row}(GQ)\subseteq\operatorname{row}(M_{\mathcal D})$. The pinned-response identity follows by
multiplying $P_0(I-B_0)=I$ and subtracting. None of these equalities characterizes a larger nonlinear or
unprofiled law fibre.

\subsection{Inference proofs}\label{app:inference-proofs}

\paragraph{Theorem~\ref{thm:instr}.}
The reconstructed mean is a smooth function of the sample mean and, when applicable, the sensor estimate. The
multivariate central limit theorem and delta method give the stated influence covariance, including the two
same-sample cross terms. Studentization on the declared estimable range gives the chi-square limit. A plugged-in
sensor without its influence contribution is outside this argument.

\paragraph{Proposition~\ref{prop:mean-concentration}.}
The stated whitened sub-Gaussian condition gives the Euclidean sample-mean bound with radius
$t_{e,n}^m(\eta)$. Under the alternative, the triangle inequality leaves a deviation larger than the null
radius whenever the population gap exceeds the sum of the null and power radii. A sample-split sensor error is
added with the displayed probability allocation.

\paragraph{Proposition~\ref{prop:covariance-adf}.}
Writing the centered quadratic score in Frobenius-isometric symmetric coordinates gives an asymptotically linear
statistic with covariance $W_e$. The influence function includes the centering correction and every
estimated nuisance derivative and overlap block stated in the proposition. In the externally known exact frame,
the full-rank studentized quadratic form is chi-square. Merely sharing an unidentified $H$ across twins is insufficient: the free-block selection and
singular subspaces need not be preserved by the residual gauge.  If a separately proved whitening map is used,
it must be common, identified, positive definite, block preserving, and its estimation influence must enter the
same joint covariance.  The proposition makes no such optional claim.

\paragraph{Proposition~\ref{prop:covariance-concentration}.}
Write
$\widehat\Sigma-\Sigma=n^{-1}\sum_i(W_iW_i^\top-\Sigma)-\bar W\bar W^\top$.
Products of sub-Gaussian coordinates are sub-exponential, so entrywise Bernstein gives the two-branch threshold.
A union bound over the $k_e$ symmetric coordinates and
$\|\mathrm{svec}(A)\|_2\le\sqrt{2k_e}\|A\|_{\max}$ give $q_{e,n}^\Sigma$ and the two displayed sample-size terms. The declared
eigenvalue envelope is needed to make the radius operational.

\paragraph{Proposition~\ref{prop:covariance-degenerate}.}
On a declared fixed-rank stratum, spectral truncation restricts the Wald statistic to the estimable range and
gives the corresponding chi-square limit.  Consistent rank selection requires a fixed population eigengap and
$O_p(n^{-1/2})$ projector error.  The hard-kernel companion is available only when the kernel is exact: sample
centering contributes $O_p(n^{-1})$, whereas a fixed kernel alternative moves the score by $\Theta(1)$, so a
threshold $t_n\downarrow0$ with $nt_n\to\infty$ has pointwise eventual power one and vanishing null rejection.
The hard-clamp formula additionally assumes independent zero-mean coordinates and a proved span equality for
the noise-square directions.  That discrete construction lies outside the paper's standing continuous-noise
class and is presented only as a separately scoped extension. Under a soft probe the actual $W_e$ determines
the rank, and neither the hard-clamp span nor uniform near-singular power may be imported.

\paragraph{Corollary~\ref{cor:aggregate-tolerance} and Proposition~\ref{prop:invariance-ceiling}.}
For a validation-independent fixed candidate, the aggregate result substitutes the mean and covariance radii
above into Corollary~\ref{cor:query-tolerance}, using the fixed block scales and the same direct-sum Euclidean
and $\mathrm{svec}$ norms as the blockwise procedures.  Acceptance is the intersection of blockwise events;
the power guarantee is therefore controlled by the largest separated block Type-II bound, not by a union of
those bounds.

For the Gaussian scale illustration, if $\Sigma_\tau=c\Sigma_\pi$ and the source covariance uses divisor
$n-1$, then
\[
c(n-1)\operatorname{tr}(\Sigma_\tau^{-1}S)\sim\chi^2_{m(n-1)}.
\]
Thus for $D=\operatorname{tr}(\Sigma_\tau^{-1}S)-m$,
\[
E[D]=\frac{m(1-c)}{c},\qquad
\operatorname{Var}(D)=\frac{2m}{c^2(n-1)},\qquad
\mathrm{SNR}^2=\frac{(n-1)m(1-c)^2}{2}.
\]
With divisor $n$, both the pivot and null centering change. The scale statistic is blind to a trace-free change
only in the stated orientation: source covariance $\Omega'=\operatorname{diag}(1+a,1-a,1)$ and comparator $I$.
Reversing the two gives $\operatorname{tr}(\Omega'^{-1})-3=2a^2/(1-a^2)>0$. The full covariance statistic sees
either nonzero change.

Finally, residual-law invariance is necessary for an invariant mechanism but cannot distinguish two
quotient-valid models whose regime-specific exact laws agree.  The passive Gaussian collision above and the
parent-dependent noise rearrangement provide existential pairs for the stated regimes.  Representative-specific
quantities such as entries of an unidentified sensor are not admissible queries.  The ceiling is therefore
regime-specific and existential, not a universal claim about every intervention library.

\small\bibliographystyle{plainnat}\bibliography{ecg-c2-references}
\end{document}